\documentclass[a4paper,11pt]{article}
\usepackage[pdftex]{graphicx}
\usepackage{graphicx,a4wide,listings}
\usepackage[intlimits]{amsmath}
\usepackage{amsxtra, amssymb,fancyhdr, amsthm,latexsym}
\usepackage[english]{babel}
\usepackage[ansinew]{inputenc}
\usepackage[bookmarks]{}
\usepackage{enumerate}
\usepackage{enumitem}
\numberwithin{equation}{section}
\numberwithin{figure}{section}

\usepackage{tikz}

\usepackage{caption}
\usepackage{subcaption}

\date{}

\begin{document}

\newtheorem{theorem}{Theorem}[section]
\newtheorem{lemma}[theorem]{Lemma}
\newtheorem{claim}[theorem]{Claim}
\newtheorem{proposition}[theorem]{Proposition}
\newtheorem{postulate}[theorem]{Postulate}
\newtheorem{definition}[theorem]{Definition}
\newtheorem{assumption}[theorem]{Assumption}
\newtheorem{corollary}[theorem]{Corollary}
\theoremstyle{definition}
\newtheorem{remark}[theorem]{Remark}

\newcommand{\BL}{{\mathrm{BL}}}
\newcommand{\TV}{{\mathrm{TV}}}
\newcommand{\FM}{{\mathrm{FM}}}
\newcommand{\pair}[2]{\left\langle #1 , #2 \right\rangle}
\newcommand{\Int}[4]{\int_{#1}^{#2}\! #3 \, #4}
\newcommand{\map}[3]{#1:#2\to#3}

\newcommand{\supp}{\operatorname*{supp}}

\newcommand{\K}{\mathcal{K}}
\newcommand{\CM}{\mathcal{M}}
\newcommand{\CMc}{\overline{\CM}}
\newcommand{\smfrac}[2]{\mbox{$\frac{#1}{#2}$}}
\newcommand{\geqs}{\geqslant}
\newcommand{\leqs}{\leqslant}
\newcommand{\nn}{{(n)}}
\newcommand{\mm}{{(m)}}

\def\R{\mathbb{R}}
\def\Rp{\mathbb{R}^+}
\def\N{\mathbb{N}}
\def\Np{\mathbb{N}^+}
\def\eps{\varepsilon}
\def\aeps{a}
\def\Wh{W_h}
\def\P{\mathcal{P}}

\title{From continuum mechanics to SPH particle systems and back: Systematic derivation and convergence}

\author{Joep H.M. Evers\thanks{Dept.~of Mathematics, Simon Fraser University, Burnaby, Canada and Dept.~of Mathematics and Statistics, Dalhousie University, Halifax, Canada (jevers@sfu.ca).}, Iason A. Zisis\thanks{CASA - Centre for Analysis, Scientific computing and Applications, Eindhoven University of Technology, The Netherlands (iason.zisis@outlook.com).}, Bas J. van der Linden\thanks{CASA - Centre for Analysis, Scientific computing and Applications, Eindhoven University of Technology, The Netherlands (b.j.v.d.linden@tue.nl).}, Manh Hong Duong\thanks{Mathematics Institute, University of Warwick, United Kingdom (m.h.duong@warwick.ac.uk).}}

\maketitle

\begin{abstract}
In this paper, we derive from the principle of least action the equation of motion for a continuous medium with regularized density field in the context of measures. The eventual equation of motion depends on the order in which regularization and the principle of least action are applied. We obtain two different equations, whose discrete counterparts coincide with the scheme used traditionally in the Smoothed Particle Hydrodynamics (SPH) numerical method (e.g. \cite{Mon05}), and with the equation treated by Di Lisio \textit{et al.}~in \cite{DiLisio}, respectively. Additionally, we prove the convergence in the Wasserstein distance of the corresponding measure-valued evolutions, moreover providing the order of convergence of the SPH method. The convergence holds for a general class of force fields, including external and internal conservative forces, friction and non-local interactions. The proof of convergence is illustrated numerically by means of one and two-dimensional examples.

\end{abstract}

\section{Introduction}
The Smoothed Particle Hydrodynamics (SPH) numerical method was initially introduced to solve the equations of astrophysical flows. In the course of time it found application to equations describing a plethora of physical processes (for its diverse applications, see \cite{Mon12}). These processes predominantly involve continua and the equations refer to systems with infinite degrees of freedom. The central idea of the SPH method is to set up a relation between the continuum and a particle system, in which the continuum is loosely considered to be the limit case in which the number of particles tends to infinity. Here, a `particle' should \textit{not} be interpreted as a physical object of any scale (like an atom, molecule or grain) but rather as a numerical entity attributed with  mass, position, velocity and other properties of the medium it represents.\\
\\
It is well-established that the classical SPH scheme can be derived formally by applying the \emph{principle of least action} 
to the particle system, where the SPH density approximation acts as a constraint; see e.g.~\cite{Mon05,Price12,Zisis14}. The importance of the particle system's Lagrangian function was already recognized in the first articles describing SPH; cf.~\cite{Mon78}. A subtlety lies in the fact that in the derivation of the SPH equations, the action of the particle system is minimized rather than the action of the continuum. The  minimization of the action at the continuum level and the subsequent discretization of the motion equation in terms of particles do not necessarily yield the same equation (at the discrete level).\\
\\
The main achievement of this paper is twofold:
\begin{itemize}
\item We introduce a systematic procedure for deriving measure-valued and particle formulations of continuum mechanics equations. We obtain two different schemes depending on the stage at which a regularization of the density is introduced. See Section~\ref{sec: derivation}.
\item We prove the convergence of both schemes using the Wasserstein distance on the space of probability measures; cf.~Section~\ref{sec: conv}.
\end{itemize}
We now describe the two parts of our paper in more detail.\\
\\
In the first part (Section \ref{sec: derivation}), we aim at clarifying the exact difference in outcome between minimizing the action of the particle system and minimizing the action at the continuum level. To achieve this, we introduce a systematic procedure consisting of the following three steps:
\begin{enumerate}[label=\textbf{\Alph*}]
\item formulation in terms of measures and, simultaneously, the regularization\footnote{The regularization of equations is an old concept, introduced by Friedrichs in 1944 \cite{Fr44}. Additionally, notice that the regularization kernels used in SPH are a special subclass of the mollifier functions used by Friedrichs \cite{Fr44}; the SPH kernels are symmetric positive mollifiers.} of the density;
\item introduction of a particle formulation;
\item application of the principle of least action.
\end{enumerate}
These three steps are introduced in more detail in Section \ref{sec: 3 proc}. It turns out that the order in which these steps are executed determines what the resulting equation is. To be more precise, the classical SPH scheme (as described e.g.~in \cite{Mon05}) is obtained, whenever the regularization of the density takes place before applying the principle of least action. That is, whenever the steps are executed in the order \ref{enum:stepA}-\ref{enum:stepB}-\ref{enum:stepC} or \ref{enum:stepA}-\ref{enum:stepC}-\ref{enum:stepB}. Both procedures are presented here; see Sections \ref{sec: eqns motion ABC} and \ref{sec: eqns motion ACB}. If we apply the principle of least action (to the action at the continuum mechanics level) before regularizing the density then we obtain a scheme that appears in Di Lisio \textit{et al.}~\cite{DiLisio} and in the recent paper \cite{Colagrossi}. However, this variant of the scheme is studied far less in literature. The procedure to obtain this scheme follows the order \ref{enum:stepC}-\ref{enum:stepA}-\ref{enum:stepB}. Its distinct characteristic is that it requires the gradient of the pressure field to be expressed analytically, while the pressure itself does not appear in the numerical scheme, in contrast to the commonly used SPH schemes. We emphasize that although both schemes arrive from the principle of least action, the latter can also be derived directly from Newtonian mechanics and introduction of the density regularization. The details of our rational derivation and the mutual relation between the two schemes have, to our knowledge, not been described in literature before.\\
\\
In any case, regularization of the density practically means that the original problem is deliberately turned into a regularized one, which is afterwards solved by means of some variant of an SPH scheme. Hence, by choosing SPH as the solution method one is automatically bound to studying a different problem than the original one at the continuum level. Thus, two questions naturally arise:
\begin{itemize}
  \item Does the solution of the regularized problem converge to the solution of the original problem?
  \item Does the particle solution of the regularized problem converge (in a certain sense) to the solution of the regularized continuum problem?
\end{itemize}
The former is out of the scope of the present study (\cite{DiLisio2} has dealt with it), while the latter is the topic of the current work's second part.\\
\\
Measure theory provides a framework to study the limiting behaviour as the number of particles goes to infinity (cf.~also e.g.~\cite{vMeurs14}). Both the particle system, and the limiting continuum setting can be formulated in terms of measures. Hence, a distance between measures is a natural tool to characterize convergence; in this work we take the Wasserstein distance on the space of probability measures. This particular distance has the advantage that it can be formulated as the \textit{infimum} over a set of joint representations (more details follow in Definitions \ref{def:coupling} and \ref{def:Wass}). This is convenient, since one can thus obtain an upper bound (needed to prove convergence!) by choosing any admissible joint representation. See also \cite{Vil09}, Chapter 6, for more discussion.\\
\\
We prove the convergence of measure-valued solutions, as the initial measure is approximated; cf.~Section~\ref{sec: conv}. The line of arguments is similar to the one followed by Di Lisio \textit{et al.}~in \cite{DiLisio}, who first employed measures in combination with the Wasserstein distance to prove the convergence of the SPH method, but the result obtained in the present work is more general. It should be mentioned that in the earlier work \cite{Oelschlaeger} convergence of the \textit{empirical measure} representing the particle system was proven, but using a different technique. Moreover, the only forces considered were mutual interactions between particles. Other approaches to obtain convergence are given e.g. by \cite{BenMoussa} using maximum local entropy estimates, \cite{Quinlan} employing estimates for the truncation error, and \cite{Raviart,Kimura}. \\
Nonetheless, the scheme treated in \cite{DiLisio} is not the aforementioned traditional scheme. Our proof applies both to traditional SPH and to the scheme covered by \cite{DiLisio}. Moreover we allow for a much more general class of force fields, including external and internal conservative forces, as well as friction and non-local interactions.\\
Previous work in the framework of measures by the authors of the current paper can be found in \cite{Duong}, where apart from the aforementioned force terms also random noise is treated. In \cite{EvHiMu_CRM,EvHiMu_JDE} measure-valued evolutions are treated in the scope of equations of motion that are first-order in time. The link between first-order and second-order models is discussed in \cite{EvFeRy}.\\
\\
The theoretical result of this paper regarding the order of convergence is supported numerically in Section \ref{sec: num} for one and two-dimensional illustrative examples, which involve different force fields.\\
\\
In Section \ref{sec: concl and future} concluding remarks are given about Sections \ref{sec: derivation} and \ref{sec: conv}. Also, some attention is given to possible future research directions.

\section{Systematic derivation of the equations of motion}\label{sec: derivation}
In this section we derive equations of motion from Hamilton's principle of least action, which involves the Lagrangian function posed in a continuum mechanics setting. We describe an explicit `recipe', hence avoiding the need to introduce approximations in an \textit{ad hoc} manner. This `recipe' consists of three building blocks (coined \ref{enum:stepA}, \ref{enum:stepB} and \ref{enum:stepC}; see Section \ref{sec: 3 proc}). The order in which these blocks are executed, influences the final outcome. As such, the systematic procedure we describe here also shows exactly how different formulations/schemes arise from the same basic principles.

\subsection{Derivation of the action in a continuous setting}\label{sec: deriv action}
Assume that for fixed time $t$ a mass density $\rho_t$ on a spatial domain $\Omega_t$ is given. We define the \textit{Lagrangian density} of our system as
\begin{equation}
\mathcal{L}(\rho;y,u):= \left(\dfrac12 |u|^2 - e(\rho(y),y)\right)\rho(y),
\end{equation}
where $y$ and $u$ are independent Eulerian coordinates, and $e$ denotes the internal energy density. To obtain the \textit{Lagrangian} $L$, we integrate $\mathcal{L}$ over the spatial domain $\Omega_t$:
\begin{equation}\label{eqn: def Lagrangian in Euler-coord}
L(t) := \Int{\Omega_t}{}{\mathcal{L}(\rho_t;y,u)}{dy}.
\end{equation}
For this integration to make sense, we assume now that $u$ is actually a velocity field, defined as a function of $t$ and $y$: $u:=u(t,y)$. Let there be a coordinate transform $\Phi_t$ such that $\Omega_t=\Phi_t(\Omega_0)$ for some initial domain $\Omega_0$. We call the family of transformations $(\Phi_t)_{t\geqs0}$ a \textit{motion mapping} and transform the integral above according to $y=\Phi_t(x)$ with $x\in\Omega_0$:
\begin{equation}\label{eqn: Lagr after coord transf}
L[\Phi](t) = \Int{\Omega_0}{}{\left(\dfrac12 |u(t,\Phi_t(x))|^2 - e(\rho_t(\Phi_t(x)),\Phi_t(x))\right)\rho_t(\Phi_t(x))\,\left|J\Phi_t(x)\right|}{dx}.
\end{equation}
The \textit{functional dependence} of $L$ on the motion mapping is indicated by explicitly including $\Phi$ in square brackets. The expression $\left|J\Phi_t\right|$ denotes the determinant of the Jacobian matrix of the transformation, consisting of the derivatives of the components of $\Phi_t$ with respect to the components of $x$. Now we assume that the density $\rho_t$ relates to the density $\rho_0$ defined on the original domain $\Omega_0$ by the same transformation $\Phi_t$, which is mathematically described by a \textit{push-forward}, $\rho_t=\Phi_{t}\#\rho_0$ (cf.~Definition \ref{def:push-forward}). In particular, the densities relate in the following way (see e.g.~\cite{Chadwick}, p.~90):
\begin{equation}\label{eqn: push-forw density transf}
\rho_0(x)=\rho_t(\Phi_t(x))\,\left|J\Phi_t(x)\right|.
\end{equation}
Combined, \eqref{eqn: Lagr after coord transf} and \eqref{eqn: push-forw density transf} yield
\begin{equation}
L[\Phi](t) = \Int{\Omega_0}{}{\left(\dfrac12 |u(t,\Phi_t(x))|^2 - e(\rho_t(\Phi_t(x)),\Phi_t(x))\right)\rho_0(x)}{dx}.
\end{equation}
In the above we fixed $t$, but obviously all arguments can be repeated for every $t$ in some interval $[0,T]$. In particular, we are interested in those motion mappings that are continuous and differentiable in time, and we wish to obtain their equation of motion. 
The introduction of the motion mapping $(\Phi_t)_{t\in[0,T]}$ has taken us from pure Eulerian coordinates in \eqref{eqn: def Lagrangian in Euler-coord} towards Lagrangian (material) coordinates in \eqref{eqn: Lagr after coord transf}. The crucial and final step to complete this procedure is now to specify what the velocity field $u$ is. In order to remain consistent with the motion mapping we introduced, we postulate the relation:
\begin{equation}\label{eqn: relation u dotPhi}
u(t,\Phi_t(x))=\dot{\Phi}_t(x) \,\,\,\text{for all }x\in\Omega_0.
\end{equation}
The velocity $u(t,\Phi_t(x))$ is the velocity at time $t$ of a material point that started in $x$ at time $0$, and -- in words -- \eqref{eqn: relation u dotPhi} means that this velocity is equal to the time derivative at time $t$ of the position $\Phi_t(x)$ of that particular material point. By connecting the Eulerian velocity $u$ to the Lagrangian velocity $\Phi_t$, we obtain the Lagrangian functional
\begin{equation}\label{eqn: Lagrangian in Phi}
L[\Phi](t) = \Int{\Omega_0}{}{\left(\dfrac12 |\dot{\Phi}_t(x))|^2 - e(\rho_t(\Phi_t(x)),\Phi_t(x))\right)\rho_0(x)}{dx}.
\end{equation}
We define the \textit{action} of our system by
\begin{equation}\label{def:action}
S[\Phi] := \Int{0}{T}{L[\Phi](t)}{dt}.
\end{equation}

\subsection{Three procedures}\label{sec: 3 proc}
The aim of this part of our paper is to derive equations of motion from the action \eqref{def:action}, by means of the Euler-Lagrange equations (we will see that these appear in different shapes). Moreover, we wish to derive these equations of motion for a particle system, which naturally induces a numerical scheme. A methodological way to go from the continuum (Section \ref{sec: deriv action}) to a particle system, is via a measure-valued formulation. Our motivation to do so is the fact that we need a framework that incorporates the `real physics', i.c. the density $\rho_t$, and an approximating particle system to establish the convergence of the particle scheme to the continuum.\\
\\
To get the transition from the continuous action \eqref{def:action} to equations of motion for the particle positions, three steps are necessary:
\begin{enumerate}[label=\textbf{\Alph*}]
\item introduction of measures: replace $\rho_t(x)dx$ by $\mu_t(dx)$ and, wherever necessary, approximate $\rho_t$ by some $\tilde{\rho}_t$ that depends on $\mu_t$;\label{enum:stepA}
\item substituting for $\mu_t$ a discrete measure $\bar{\mu}^n_t = \sum_i m_i \delta_{x_i(t)}$;\label{enum:stepB}
\item Derive the Euler-Lagrange equations (either classically or in variational sense).\label{enum:stepC}
\end{enumerate}
The steps are here described in a somewhat simplistic and unprecise way; their true meaning will become clear in Sections \ref{sec: eqns motion ABC}, \ref{sec: eqns motion ACB} and \ref{sec: eqns motion CAB}. Step \ref{enum:stepA} takes us to a regularized version of the problem, which is a problem different from the original one. Step \ref{enum:stepB} cannot happen before \ref{enum:stepA}, but we have the freedom to choose the further ordering. This gives rise to three different derivations:
\begin{description}
\item[\ref{enum:stepA}\ref{enum:stepB}\ref{enum:stepC} ] this procedure discretizes the Lagrangian and derives the corresponding equations of motion afterwards; see Section \ref{sec: eqns motion ABC}.
\item[\ref{enum:stepA}\ref{enum:stepC}\ref{enum:stepB} ] this procedure derives the equations of motion from the measure-valued Lagrangian and discretizes these equations afterwards; see Section \ref{sec: eqns motion ACB}.
\item[\ref{enum:stepC}\ref{enum:stepA}\ref{enum:stepB} ] this procedure derives the equations of motion from the continuum Lagrangian, writes them in measure-valued form and discretizes afterwards; see Section \ref{sec: eqns motion CAB}.
\end{description}
\newcommand{\procABC}{{\ref{enum:stepA}\ref{enum:stepB}\ref{enum:stepC} }}
\newcommand{\procACB}{{\ref{enum:stepA}\ref{enum:stepC}\ref{enum:stepB} }}
\newcommand{\procCAB}{{\ref{enum:stepC}\ref{enum:stepA}\ref{enum:stepB} }}
Procedures \procABC and \procACB eventually yield the same particle scheme. This is the scheme traditionally used in the SPH community (cf.~\cite{Mon05}). Procedure \procCAB is the one that yields the equations used in \cite{DiLisio} and \cite{Colagrossi}.
\pagebreak
\subsection{Equations of motion via the route \procABC}\label{sec: eqns motion ABC}
\subsubsection{Step \ref{enum:stepA}}\label{sec:proc ABC step A}
In Section \ref{sec: deriv action} we introduced (for each $t$) the density $\rho_t$ as the push-forward of the initial density $\rho_0$ under the mapping $\Phi_t$. In this section we lift the evolution of $\rho_t$ to the space of (time-dependent) measures. Let $\mu_0$ and $\mu_t$ be the measures associated to the densities $\rho_0$ and $\rho_t$. Hence, $\mu_t=\Phi_t\#\mu_0$. In \eqref{eqn: Lagrangian in Phi}--\eqref{def:action} we can substitute $\rho_0(x)dx$ by $\mu_0(dx)$. Afterwards, there is one more aspect that we need to `repair' before we are completely in a measure formulation. The internal energy density $e$ depends on $\rho_t$ itself, via pointwise evaluation at $\Phi_t(x)$. An approximation of $\rho_t$ is needed to obtain a general expression that is even well-defined for measures that have no density (w.r.t.~the Lebesgue measure). We propose to introduce a regularization via convolution
\begin{equation}\label{eqn: def tilde rho}
\tilde{\rho}_t(\xi):= \left(\Wh*\mu_t\right)(\xi) = \Int{\Omega_t}{}{\Wh(\xi-y)}{\mu_t(dy)},
\end{equation}
for all $\xi\in\R^d$. Here, the smoothing function $\Wh$ is nonnegative and even (so that it obtains an odd gradient, an effect which is used later in the derivation of the equations), $h$ is a small parameter, and $\Wh\rightharpoonup \delta_0$ in the narrow topology as $h\to0$ (i.e.~tested against bounded continuous functions). A typical example is the Gaussian with zero mean and variance $h^2/2$. If $\mu_t$ has a density $\rho_t$ then the convergence $\tilde{\rho}_t\to\rho_t$ holds \textit{in some sense and under certain mathematical conditions}. E.g. if $\rho_t$ is continuous and bounded, then by definition of $\Wh\rightharpoonup \delta_0$, $\tilde{\rho}_t(\xi)$ converges to $\rho_t(\xi)$ for all $\xi$. In any case, the convolution regularizes the solution, introducing an artificial `density' $\tilde{\rho}_t$, such that pointwise evaluation and the gradient are defined even when $\rho_t$ does not exist or is not differentiable. Note that, $\tilde{\rho}_t$ also depends on $h$, but in this work we do not consider the limit $h\to0$, therefore for simplicity of notation, we leave out $h$ in $\tilde{\rho}_t$. However, we stick to the subscript $h$ in $\Wh$ in agreement with the common notation in SPH literature.\\
\\
Note that $\tilde{\rho}_t$ can also be written as
\begin{equation}
\tilde{\rho}_t(\xi) = \Int{\Omega_0}{}{\Wh(\xi-\Phi_t(x))}{\mu_0(dx)},
\end{equation}
by definition of the push-forward. Hence, we should keep in mind that $\tilde{\rho}$ has either a functional dependence on $\mu_t$, or an extra dependence on $\Phi_t(\cdot)$ (depending on which formulation we choose), but we do not write this dependence explicitly.\\
\\
In $e$, we substitute $\tilde{\rho}_t$ for $\rho_t$ in the sequel and redefine the Lagrangian (in a measure-formulation) such that the action becomes
\begin{equation}\label{eqn: action measure}
S[\Phi] = \Int{0}{T}{L[\Phi](t)}{dt} = \Int{0}{T}{\Int{\Omega_0}{}{\left(\dfrac12 |\dot{\Phi}_t(x)|^2 - e(\tilde{\rho}_t(\Phi_t(x)),\Phi_t(x))\right)}{\mu_0(dx)}}{dt}.
\end{equation}
The new, generalized formulation in terms of measures allows us to consider more types of solutions, simply by allowing for more general initial conditions. This is exactly what we exploit in the following step via a particle approximation.

\subsubsection{Step \ref{enum:stepB}}
In this step, we substitute for $\mu_0$ a discrete measure of the form $\bar{\mu}^n_0 = \sum_{i=1}^n m_i \delta_{x_{i,0}}$. Under push-forward, the measure remains a discrete measure with positions of the Diracs $\{x_i(t)\}$ evolving under the motion mapping: $x_i(t)=\Phi_t(x_{i,0})$. We emphasize that the equation for $\Phi_t$ is yet unknown and is to be derived in the next step.\\
\\
The Lagrangian takes the form
\begin{equation}\label{eqn: Lagr discrete measure}
L[\Phi](t) = \sum_{i=1}^n m_i\,\left(\dfrac12 |\dot{x}_i(t)|^2 - e(\tilde{\rho}_t(x_i(t)),x_i(t))\right),
\end{equation}
with
\begin{equation}
\tilde{\rho}_t(x_i(t)) = \sum_{j=1}^n m_j\,\Wh(x_i(t)-x_j(t)).
\end{equation}
In the literature of SPH, particles of the same mass are employed for the modeling of the flow of a single fluid. In that case, the term $m_i$ corresponds to $1/n$. On the other hand, multiphase media of piecewise continuous mass density can be modeled with the use of particles of different masses \cite{Mon05,Zisis14}. For that reason, we adopt the general case of (in principle) unequal masses $m_i$.

\subsubsection{Step \ref{enum:stepC}}
The equations of motion are obtained via the `classical' Euler-Lagrange equations, see (1.57) in \cite{Gol01}, applied to the Lagrangian
\begin{equation}
L[\Phi](t) = \sum_{i=1}^n m_i\,\left(\dfrac12 |u_i|^2 - e\left(\sum_{j=1}^n m_j\,\Wh(y_i-y_j),y_i\right)\right),
\end{equation}
cf.~\eqref{eqn: Lagr discrete measure}. In the presence of nonconservative forces (cf.~p.~23 in \cite{Gol01}) the corresponding Euler-Lagrange equations are
\begin{multline}\label{eqn:EL part}
\dfrac{d}{dt}\left(\nabla_{u_k}L\left|_{(y_i,u_i)=(\Phi_t(x_{i,0}),\dot{\Phi}_t(x_{i,0}))}\right)-\nabla_{y_k}L\right|_{(y_i,u_i)=(\Phi_t(x_{i,0}),\dot{\Phi}_t(x_{i,0}))}\\
= m_k\,q[\bar{\mu}^n_t](\Phi_t(x_{k,0}),\dot{\Phi}_t(x_{k,0})),
\end{multline}
for each $k\in\{1,\ldots,n\}$, where $q$ is the force density (per unit mass) of nonconservative forces. The functional dependence in square brackets denotes that $q$ incorporates a nonlocal interaction term. More details will follow later; cf.~\eqref{eqn: q = eta + K}. The subscript ``$(y_i,u_i)=(\Phi_t(x_{i,0}),\dot{\Phi}_t(x_{i,0}))$'' should be read as performing this substitution for all $i\in\{1,\ldots,n\}$.\\
\\
After calculating the derivatives $\nabla_{u_k}$ and $\nabla_{y_k}$ in \eqref{eqn:EL part}, we obtain
\begin{multline}\label{eqn:EL part full}
m_k\dfrac{d}{dt}\dot{\Phi}_t(x_{k,0}) = \\
- m_k\,\dfrac{\partial e}{\partial \rho}\left(\sum_{j=1}^n m_j\,\Wh(\Phi_t(x_{k,0})-\Phi_t(x_{j,0})),\Phi_t(x_{k,0})\right)\sum_{j=1}^n m_j\,\nabla \Wh(\Phi_t(x_{k,0})-\Phi_t(x_{j,0}))\\
+ \sum_{i=1}^n m_i\,\dfrac{\partial e}{\partial \rho}\left(\sum_{j=1}^n m_j\,\Wh(\Phi_t(x_{i,0})-\Phi_t(x_{j,0})),\Phi_t(x_{i,0})\right)m_k\,\nabla \Wh(\Phi_t(x_{i,0})-\Phi_t(x_{k,0})) \\
- m_k\,\nabla_y e\left(\sum_{j=1}^n m_j\,\Wh(\Phi_t(x_{k,0})-\Phi_t(x_{j,0})),\Phi_t(x_{k,0})\right) + m_k\,q[\bar{\mu}^n_t](\Phi_t(x_{k,0}),\dot{\Phi}_t(x_{k,0})).
\end{multline}
We denote by $\nabla_y e$ the gradient of $e$ only in the explicit spatial coordinate; that is, the second variable of $e$. We divide all terms by $m_k$ (which is nonzero without loss of generality). If in the second line we take $\partial e/\partial\rho$ inside the sum and we use in the third line that $\nabla \Wh$ is an odd function, then the corresponding terms in \eqref{eqn:EL part full} can be combined, and we obtain
\begin{multline}
\ddot{\Phi}_t(x_{k,0}) =\\
-\sum_{i=1}^n m_i\,\nabla \Wh(\Phi_t(x_{k,0})-\Phi_t(x_{i,0}))\,\left[\dfrac{\partial e}{\partial \rho}\left(\tilde{\rho}_t(\Phi_t(x_{k,0})),\Phi_t(x_{k,0})\right)+\dfrac{\partial e}{\partial \rho}\left(\tilde{\rho}_t(\Phi_t(x_{i,0})),\Phi_t(x_{i,0})\right)\right] \\
- \nabla_y e\left(\tilde{\rho}_t(\Phi_t(x_{k,0})),\Phi_t(x_{k,0})\right) + q[\bar{\mu}^n_t](\Phi_t(x_{k,0}),\dot{\Phi}_t(x_{k,0})),\label{eqn: eqn motion discrete ABC}
\end{multline}
for each $k\in\{1,\ldots,n\}$. For brevity of notation, we use $\tilde{\rho}$ again in the argument of $e$.
\subsection{Equations of motion via the route \procACB}\label{sec: eqns motion ACB}
\subsubsection{Step \ref{enum:stepA}}
This step is exactly the same as in Section \ref{sec:proc ABC step A}.
\subsubsection{Step \ref{enum:stepC}}\label{sec:proc ACB step C}
We start from the action given in \eqref{def:action}. Instead of using the classical Euler-Lagrange equations, we employ here a generalized form of the principle of least action (see p.~127 and Section 4.4 of \cite{Berdichevsky}):
\begin{equation}\label{eqn: EL variational}
S'[\Phi](\Psi)= -Q[\Phi](\Psi),
\end{equation}
for all test functions $\Psi\in C_c^\infty((0,T);C_c^\infty(\Omega_0;\R^d))$. Here, $S'[\Phi](\Psi)$ denotes the variational derivative of $S$ in the direction of $\Psi$, and $Q[\Phi](\Psi)$ is the work done along $\Psi$. It is defined as
\begin{equation}\label{eqn: Q}
Q[\Phi](\Psi) := \Int{0}{T}{\Int{\Omega_0}{}{q[\mu_t](\Phi_t(x),\dot{\Phi}_t(x))\cdot\Psi_t(x)}{\mu_0(dx)}}{dt},
\end{equation}
where $q$ is the force density as in \eqref{eqn:EL part}. For $S'$ we have:
\begin{equation}
\nonumber S'[\Phi](\Psi) := \left.\dfrac{d}{d\eps}S[\Phi+\eps\Psi]\right|_{\eps=0}.\label{eqn: var der S}
\end{equation}
Note that
\begin{multline}
\dfrac{d}{d\eps} \left[ e\left(\Int{\Omega_0}{}{\Wh(\Phi_t(x)+\eps\Psi_t(x)-\Phi_t(y)-\eps\Psi_t(y))}{\mu_0(dy)},\Phi_t(x)+\eps\Psi_t(x)\right) \right]_{\eps=0} \label{eqn:var der e}\\
 =\, \dfrac{\partial e}{\partial \rho}\left(\Int{\Omega_0}{}{\Wh(\Phi_t(x)-\Phi_t(y)}{\mu_0(dy)},\Phi_t(x)\right)
\Int{\Omega_0}{}{\nabla \Wh(\Phi_t(x)-\Phi_t(y))\cdot (\Psi_t(x)-\Psi_t(y))}{\mu_0(dy)}\\
+ \nabla_y e\left(\Int{\Omega_0}{}{\Wh(\Phi_t(x)-\Phi_t(y)}{\mu_0(dy)},\Phi_t(x)\right)\cdot\Psi_t(x).
\end{multline}
To avoid lengthy notation, we denote here by $e'[\Phi](\Psi)(x)$ the expression in \eqref{eqn:var der e}. The variational derivative of $S$ can be expressed as:
\begin{align}
S'[\Phi](\Psi) =& \Int{0}{T}{\Int{\Omega_0}{}{\left(\dot{\Phi}_t(x)\cdot\dot{\Psi}_t(x) - e'[\Phi](\Psi)(x)\right)}{\mu_0(dx)}}{dt} \label{eqn:var der S} \\
\nonumber =& \Int{0}{T}{\Int{\Omega_0}{}{\left(-\ddot{\Phi}_t(x)\cdot\Psi_t(x) - e'[\Phi](\Psi)(x)\right)}{\mu_0(dx)}}{dt},
\end{align}
where the last step follows from integration by parts with respect to the time variable. The boundary terms disappear because $\Psi$ has compact support within $(0,T)$.\\
\\
We rewrite the part involving $\Psi_t(y)$ in \eqref{eqn:var der e} as follows:
\begin{multline}
\Int{0}{T}{\Int{\Omega_0}{}{-\dfrac{\partial e}{\partial \rho}\left(\tilde{\rho}_t(\Phi_t(x)),\Phi_t(x)\right)\,\Int{\Omega_0}{}{\nabla \Wh(\Phi_t(x)-\Phi_t(y))\cdot \Psi_t(y)}{\mu_0(dy)}}{\mu_0(dx)}}{dt}\\
=\Int{0}{T}{\Int{\Omega_0}{}{\Int{\Omega_0}{}{\dfrac{\partial e}{\partial \rho}\left(\tilde{\rho}_t(\Phi_t(y)),\Phi_t(y)\right)\,\nabla \Wh(\Phi_t(x)-\Phi_t(y))}{\mu_0(dy)}\cdot \Psi_t(x)}{\mu_0(dx)}}{dt},\label{eqn:change order int}
\end{multline}
by subsequently interchanging the order of integration, using that the function $\nabla W$ is odd, and replacing $x$ by $y$ and \textit{vice versa}. A combination of \eqref{eqn: EL variational}, \eqref{eqn: Q}, \eqref{eqn:var der e}, \eqref{eqn:var der S} and \eqref{eqn:change order int} yields for $S'[\Phi](\Psi)+Q[\Phi](\Psi)$ an integral of the form
\begin{equation}
\Int{0}{T}{\Int{\Omega_0}{}{[\ldots]\cdot\Psi_t(x)}{\mu_0(dx)}}{dt},
\end{equation}
where we deliberately do not explicitly write the integrand in square brackets. Since this integral equals $0$ for all $\Psi\in C_c^\infty((0,T);C_c^\infty(\Omega_0;\R^d))$ -- cf.~\eqref{eqn: EL variational} -- the theorem of du Bois-Reymond yields that the integrand should vanish for almost all $t\in[0,T]$ and for $\mu_0$-almost every $x$. Hence, we obtain
\begin{multline}
\ddot{\Phi}_t(x) = -\Int{\Omega_0}{}{\nabla \Wh(\Phi_t(x)-\Phi_t(y))\left[\dfrac{\partial e}{\partial \rho}\left(\tilde{\rho}_t(\Phi_t(x)),\Phi_t(x)\right)+\dfrac{\partial e}{\partial \rho}\left(\tilde{\rho}_t(\Phi_t(y)),\Phi_t(y)\right)   \right]}{\mu_0(dy)}\\
- \nabla_y e\left(\tilde{\rho}_t(\Phi_t(x)),\Phi_t(x)\right)  + q[\mu_t](\Phi_t(x),\dot{\Phi}_t(x)).\label{eqn: eqn motion ACB}
\end{multline}
\subsubsection{Step \ref{enum:stepB}}
The transition to a particle system takes place by substitution of $\bar{\mu}^n_0=\sum_{i=1}^n m_i\,\delta_{x_{i,0}}$ for $\mu_0$ in \eqref{eqn: eqn motion ACB}. Moreover, in $\tilde{\rho}_t$ and $q$ we replace $\mu_t$ by $\bar{\mu}^n_t:=\Phi_t\#\bar{\mu}^n_0$. Note that, after substitution, \eqref{eqn: eqn motion ACB} holds $\bar{\mu}^n_0$-a.e.~and should therefore (only) be evaluated at $x=x_{k,0}$ for all $k\in\{1,\ldots,n\}$. We obtain exactly \eqref{eqn: eqn motion discrete ABC}.

\subsection{Equations of motion via the route \procCAB}\label{sec: eqns motion CAB}
\subsubsection{Step \ref{enum:stepC}}
At the continuum level, deriving the Euler-Lagrange equations resembles considerably what was done in Section \ref{sec:proc ACB step C}. Note however that the action as defined in \eqref{eqn: Lagrangian in Phi}--\eqref{def:action} is used. In \eqref{eqn: Lagrangian in Phi} $\rho_t(\Phi_t(x))$ occurs. The dependence on $\Phi_t$ that is explicitly written down, corresponds to the position at which $\rho_t$ is evaluated. However, if $\Phi_t$ is varied, also the function $\rho_t$ itself changes. This is somewhat confusing, as this is an implicit, `hidden' dependence of $\rho_t$ on the motion mapping $\Phi_t$. However, the exact relation is given by \eqref{eqn: push-forw density transf}, which we therefore substitute in \eqref{eqn: Lagrangian in Phi}. The variational derivative becomes
\begin{multline}
\nonumber S'[\Phi](\Psi) =\\ \left.\dfrac{d}{d\eps}\left(\Int{0}{T}{\Int{\Omega_0}{}{\left(\dfrac12 |\dot{\Phi}_t(x)+\eps\,\dot{\Psi}_t(x)|^2 - e\left(\dfrac{\rho_0(x)}{\left|J(\Phi_t+\eps\,\Psi_t)(x)\right|},\Phi_t(x)+\eps\,\Psi_t(x)\right)\right)\,\rho_0(x)}{dx}}{dt}\right)\right|_{\eps=0},
\end{multline}
cf.~\eqref{eqn: var der S}. Some effort is needed to deal with the $\eps$-dependence in the Jacobian matrix. We refer here to Section 2 of \cite{Seliger}, where the equation of motion is derived from the action, for the case where $e$ has no explicit dependence on the spatial coordinate; i.e. $e=e(\rho)$. The determinant of the Jacobian matrix is a polynomial of the entries of that matrix. The basic idea in \cite{Seliger} is that the chain rule has to be applied with respect to every element of the Jacobian matrix. To avoid having to introduce a considerable amount of extra notation, we only state the result of \cite{Seliger} here:
\begin{align}\label{eqn:eqn motion Seliger}
\ddot{\Phi}_t(x) =&\left.-\dfrac{1}{\rho_t}\nabla\left(\rho_t^2\,\dfrac{\partial e}{\partial\rho}\right)\right|_{\Phi_t(x)}\\
\nonumber =& -\left(2\dfrac{\partial e}{\partial\rho}(\rho_t(\Phi_t(x))) + \rho_t(\Phi_t(x))\dfrac{\partial^2 e}{\partial\rho^2}(\rho_t(\Phi_t(x)))\right) \nabla\rho_t(\Phi_t(x)).
\end{align}
On the right-hand side the gradient of the pressure $P$ appears, due to the thermodynamic relation $\partial e/\partial \rho=P/\rho^2$. The reader should note that the notation used in \cite{Seliger} differs substantially from ours, but that the philosophy of deriving the equations of motion is the same.\footnote{Another interesting observation in \cite{Seliger} is that the Lagrangian density -- when formulated in terms of Eulerian coordinates -- is just the pressure $P$.}\\
\\
If $e=e(\rho,y)$, and moreover, we include nonconservative forces, then instead of \eqref{eqn:eqn motion Seliger} we obtain
\begin{multline}
\ddot{\Phi}_t(x)= -\left(2\dfrac{\partial e}{\partial\rho}(\rho_t(\Phi_t(x)),\Phi_t(x)) + \rho_t(\Phi_t(x))\dfrac{\partial^2 e}{\partial\rho^2}(\rho_t(\Phi_t(x)),\Phi_t(x))\right) \nabla\rho_t(\Phi_t(x))\\
- \nabla_y e\left(\rho_t(\Phi_t(x)),\Phi_t(x)\right)  + q[\rho_t](\Phi_t(x),\dot{\Phi}_t(x)).\label{eqn: eqn motion Seliger with extra terms}
\end{multline}
The additional terms follow from similar steps as the ones leading to \eqref{eqn: eqn motion ACB}. We omit further details. Note that, in correspondence with $q$ as introduced before, the dependence on $\rho_t$ in square brackets indicates the presence of a nonlocal term; cf.~\eqref{eqn: q = eta + K}. In the next step, this will become a dependence on the measure $\mu_t$ like before.

\subsubsection{Step \ref{enum:stepA}}
In this step, we formulate \eqref{eqn: eqn motion Seliger with extra terms} in terms of measures. The only place where the measure $\mu_t$ can be incorporated directly, is in the nonconservative force density. We write $q[\mu_t](\Phi_t(x),\dot{\Phi}_t(x))$ instead of $q[\rho_t](\Phi_t(x),\dot{\Phi}_t(x))$. All the other occurrences of $\rho_t$ in \eqref{eqn: eqn motion Seliger with extra terms} we approximate by $\tilde{\rho}_t$ as defined in \eqref{eqn: def tilde rho}. We obtain
\begin{multline}
\ddot{\Phi}_t(x)= -\left(2\dfrac{\partial e}{\partial\rho}(\tilde{\rho}_t(\Phi_t(x)),\Phi_t(x)) + \tilde{\rho}_t(\Phi_t(x))\dfrac{\partial^2 e}{\partial\rho^2}(\tilde{\rho}_t(\Phi_t(x)),\Phi_t(x))\right) \nabla\tilde{\rho}_t(\Phi_t(x))\\
- \nabla_y e\left(\tilde{\rho}_t(\Phi_t(x)),\Phi_t(x)\right)  + q[\mu_t](\Phi_t(x),\dot{\Phi}_t(x)).\label{eqn: eqn motion CAB}
\end{multline}

\subsubsection{Step \ref{enum:stepB}}
We take $\bar{\mu}^n_0:=\sum_{i=1}^n m_i\,\delta_{x_{i,0}}$ and replace $\mu_t$ by $\bar{\mu}^n_t:=\Phi_t\#\bar{\mu}^n_0$ in $\tilde{\rho}_t$ and $q$ that appear in \eqref{eqn: eqn motion CAB}. We evaluate the resulting equation at $x=x_{k,0}$ for all $k\in\{1,\ldots,n\}$ to obtain
\begin{multline}
\ddot{\Phi}_t(x_{k,0})=\\ -\left(2\dfrac{\partial e}{\partial\rho}(\tilde{\rho}_t(\Phi_t(x_{k,0})),\Phi_t(x_{k,0})) + \tilde{\rho}_t(\Phi_t(x_{k,0}))\dfrac{\partial^2 e}{\partial\rho^2}(\tilde{\rho}_t(\Phi_t(x_{k,0})),\Phi_t(x_{k,0}))\right) \nabla\tilde{\rho}_t(\Phi_t(x_{k,0}))\\
- \nabla_y e\left(\tilde{\rho}_t(\Phi_t(x_{k,0})),\Phi_t(x_{k,0})\right)  + q[\bar{\mu}^n_t](\Phi_t(x_{k,0}),\dot{\Phi}_t(x_{k,0})),\label{eqn: eqn motion discrete CAB}
\end{multline}
where each appearance of $\tilde{\rho}_t$ denotes a sum over all particle positions. Namely,
\begin{align}
\tilde{\rho}_t(\Phi_t(x_{k,0})) &= \sum_{j=1}^n m_j\,\Wh(\Phi_t(x_{k,0})-\Phi_t(x_{j,0})),\,\,\,\text{and}\\
\nabla\tilde{\rho}_t(\Phi_t(x_{k,0})) &= \sum_{j=1}^n m_j\,\nabla\Wh(\Phi_t(x_{k,0})-\Phi_t(x_{j,0})).
\end{align}
\subsection{Comparison of the resulting equations \eqref{eqn: eqn motion discrete ABC} and \eqref{eqn: eqn motion discrete CAB}}\label{sec: compare eqns motion}
Procedures \procABC and \procACB yield the same equations of motion, namely \eqref{eqn: eqn motion discrete ABC}. As anticipated already in Section \ref{sec: 3 proc}, the equation resulting from Procedure \procCAB is different; see \eqref{eqn: eqn motion discrete CAB}. This difference between the two final equations arose because we introduced the regularization via $\tilde{\rho}$ at different stages. In fact, \eqref{eqn: eqn motion discrete ABC} contains an extra regularization in space, as we will show now.\\
\\
Note that only the term involving $\partial e/\partial\rho$ and $\partial^2 e/\partial\rho^2$ is different. In \eqref{eqn: eqn motion discrete ABC}, we have
\begin{equation}
\nonumber -\sum_{i=1}^n m_i\,\nabla \Wh(\Phi_t(x_{k,0})-\Phi_t(x_{i,0}))\,\left[\dfrac{\partial e}{\partial \rho}\left(\tilde{\rho}_t(\Phi_t(x_{k,0})),\Phi_t(x_{k,0})\right)+\dfrac{\partial e}{\partial \rho}\left(\tilde{\rho}_t(\Phi_t(x_{i,0})),\Phi_t(x_{i,0})\right)\right],
\end{equation}
while the corresponding part in \eqref{eqn: eqn motion discrete CAB} is
\begin{equation}
\nonumber -\left(2\dfrac{\partial e}{\partial\rho}(\tilde{\rho}_t(\Phi_t(x_{k,0})),\Phi_t(x_{k,0})) + \tilde{\rho}_t(\Phi_t(x_{k,0}))\dfrac{\partial^2 e}{\partial\rho^2}(\tilde{\rho}_t(\Phi_t(x_{k,0})),\Phi_t(x_{k,0}))\right) \nabla\tilde{\rho}_t(\Phi_t(x_{k,0})).
\end{equation}
Note that both of them contain a part $-\dfrac{\partial e}{\partial\rho}(\tilde{\rho}_t(\Phi_t(x_{k,0})),\Phi_t(x_{k,0}))\,\nabla\tilde{\rho}_t(\Phi_t(x_{k,0}))$, hence let us consider in \eqref{eqn: eqn motion discrete CAB} only
\begin{multline}\label{eqn: gradient rho de/drho}
-\left(\dfrac{\partial e}{\partial\rho}(\tilde{\rho}_t(\Phi_t(x_{k,0})),\Phi_t(x_{k,0})) + \tilde{\rho}_t(\Phi_t(x_{k,0}))\dfrac{\partial^2 e}{\partial\rho^2}(\tilde{\rho}_t(\Phi_t(x_{k,0})),\Phi_t(x_{k,0}))\right) \nabla\tilde{\rho}_t(\Phi_t(x_{k,0}))\\
 = -\nabla\left( \tilde{\rho}_t(\Phi_t(x_{k,0}))\dfrac{\partial e}{\partial\rho}(\tilde{\rho}_t(\Phi_t(x_{k,0})),\Phi_t(x_{k,0})) \right).
\end{multline}
To obtain this equality, we have assumed that $\nabla_y \partial e/\partial \rho \equiv 0$; this assumption anticipates the choice we make in \eqref{eqn: explicit e = V + barF}. Let us even go back one more step and consider this term before the introduction of $\tilde{\rho}$, i.e.~as in \eqref{eqn: eqn motion Seliger with extra terms}. To see how this term relates to the corresponding one in \eqref{eqn: eqn motion discrete ABC}, we take the convolution with $\Wh$, and proceed as follows:
\begin{multline}
\nonumber-\Int{\Omega_t}{}{\Wh(\xi-y)\nabla_y\left( \rho_t(y)\dfrac{\partial e}{\partial\rho}(\rho_t(y),y) \right)}{dy} = \Int{\Omega_t}{}{\nabla_y\Wh(\xi-y)\,\rho_t(y)\,\dfrac{\partial e}{\partial\rho}(\rho_t(y),y)}{dy}\\
\nonumber= -\Int{\Omega_0}{}{\nabla\Wh(\xi-\Phi_t(y))\,\dfrac{\partial e}{\partial\rho}(\rho_t(\Phi_t(y)),\Phi_t(y))\,\rho_t(\Phi_t(y))|J\Phi_t(y)|}{dy}\\
= -\Int{\Omega_0}{}{\nabla\Wh(\xi-\Phi_t(y))\,\dfrac{\partial e}{\partial\rho}(\rho_t(\Phi_t(y)),\Phi_t(y))\,\rho_0(y)}{dy}.
\end{multline}
In the first step, we performed integration by parts, with vanishing boundary terms on $\partial\Omega_t$. This is because $\Omega_t=\supp\rho_t$ and hence $\rho_t$ vanishes on its boundary. Now replace $\rho_0(y)dy$ by $\mu_0(dy)$ and approximate $\rho_t$ by $\tilde{\rho}_t$. Take $\mu_0:=\sum_{i=1}^n m_i\,\delta_{x_{i,0}}$ and evaluate at $\xi=\Phi_t(x_{k,0})$ and obtain
\begin{equation}
\nonumber -\sum_{i=1}^n m_i\,\nabla \Wh(\Phi_t(x_{k,0})-\Phi_t(x_{i,0}))\,\dfrac{\partial e}{\partial \rho}\left(\tilde{\rho}_t(\Phi_t(x_{i,0})),\Phi_t(x_{i,0})\right).
\end{equation}
This expression exactly appears in \eqref{eqn: eqn motion discrete ABC}. To summarize: the connection between \eqref{eqn: eqn motion discrete ABC} and \eqref{eqn: eqn motion discrete CAB} is that in the former during the derivation procedure an extra regularization in space was introduced for \textit{a part of} the right-hand side. Note the connection with the following case: consider the Fr\'echet derivative $\mathcal{E}'$, based on the $L^2$ inner product, of some energy $\mathcal{E}=\mathcal{E}(\rho)$. Define a second energy $\bar{\mathcal{E}}$ by $\bar{\mathcal{E}}(\rho):=\mathcal{E}(W_h*\rho)$. Then $\bar{\mathcal{E}}'(\rho)=W_h*\mathcal{E}'(W_h*\rho)$, which also contains an extra regularization. In this paper we treat a special case of the general energy $\mathcal{E}$.\\
\\
Note that, if we only consider the part involving $\partial e/\partial \rho$, \eqref{eqn: eqn motion discrete ABC} is the same as Equation (3.8) in \cite{Mon05}. The notation used therein shows the direct dependence on the pressure. In Equation (3.5) of \cite{Mon05}, the equivalent of \eqref{eqn: eqn motion discrete CAB} is given. The reason why \eqref{eqn: eqn motion discrete ABC} is the one traditionally used in the SPH community is given in \cite{Mon05}: it does conserve linear and angular momentum exactly, as opposed to \eqref{eqn: eqn motion discrete CAB}. Having derived the schemes, we are now able also to elaborate on the remark already made in the introduction: \eqref{eqn: eqn motion discrete CAB} ``requires the gradient of the pressure field to be expressed analytically, while the pressure itself does not appear in the numerical scheme". The first part on the right-hand side of \eqref{eqn: eqn motion discrete CAB} is -- anticipating \eqref{eqn: defs F0 F1} -- of the form $-\dfrac{1}{\tilde{\rho}}\,\dfrac{d}{d\rho}\left(\tilde{\rho}^2 \bar{F}'(\tilde{\rho})\right)\nabla\tilde{\rho} =-\dfrac{1}{\tilde{\rho}}\,\dfrac{d}{d\rho}\left(P(\tilde{\rho})\right)\nabla\tilde{\rho}$. Hence we need an analytical expression for $\dfrac{d}{d\rho}\left(P(\tilde{\rho})\right)$.

\subsection{Measure-valued formulation}
In Sections \ref{sec: eqns motion ABC}, \ref{sec: eqns motion ACB} and \ref{sec: eqns motion CAB} we derived particle-based schemes. To establish their convergence (as $n\to\infty$) we use a measure-valued formulation. Such formulation incorporates both the limit and the approximating sequence. Hence, we focus on the measure-formulations \eqref{eqn: eqn motion ACB} and \eqref{eqn: eqn motion CAB}, without the specific choice $\mu_0=\bar{\mu}^n_0$. Our convergence proof is applicable to a class of approximating measures that is much broader than just sums of Dirac deltas. The SPH-inspired particle approach is a special case; see Corollary \ref{cor: particle system}.\\
\\
Although \eqref{eqn: eqn motion ACB} and \eqref{eqn: eqn motion CAB} are different (cf.~Section \ref{sec: compare eqns motion}), we wish to establish the convergence proof for both formulations simultaneously. Hence, we introduce a switching parameter $\theta\in\{0,1\}$ to unify both variants in a single equation of motion. First, we assume that $e$ is of the form
\begin{equation}
e(\rho,y) := V(y) + \bar{F}(\rho),\label{eqn: explicit e = V + barF}\\
\end{equation}
in agreement with the remark we already made underneath \eqref{eqn: gradient rho de/drho}. Note that $\partial e/\partial \rho = \bar{F}'$ and $\nabla_y e=\nabla V$. Here, $V\in C^2_b(\R^d;\R)$ describes the portion of potential energy which is due to a gravitational or magnetic field and $\bar{F}\in C^2(\Rp;\R)$, where $\Rp:=(0,\infty)$ the potential energy due to the thermodynamics of the medium under consideration. This decomposition of $e$ is typical for an ideal medium, such as a compressible inviscid fluid. Note moreover that this is a common modeling assumption in the derivation of the SPH equations for a system of particles \cite{Mon05}. We introduce an auxiliary function $F_\theta$, $\theta\in\{0,1\}$, that is defined by
\begin{equation}\label{eqn: defs F0 F1}
F_0(\rho) := \dfrac{1}{\rho}\,\dfrac{d}{d\rho}\left(\rho^2 \bar{F}'(\rho)\right),\,\,\,\,\,\text{and}\,\,\,\,\,F_1(\rho) := \bar{F}'(\rho).
\end{equation}
We choose $q$ to be of the form
\begin{equation}
q[\mu](y,u) := -\eta(y)\,u + (K*\mu)(y),\label{eqn: q = eta + K}
\end{equation}
with $\eta\in C^1_b(\R^d;\Rp)$ and $K\in C^1_b(\R^d;\R^d)$. The $K$-term describes non-local interactions within the system, while the $\eta$-term is a viscous term. We use $-\eta\cdot u$, which is a simplified version of the usual viscous term in SPH that (also) involves $\Delta W_h * u$; see \cite{Mon05}.\\
\\
We assign the value $\theta=0$ to the formulation in \eqref{eqn: eqn motion CAB}, and $\theta=1$ to \eqref{eqn: eqn motion ACB}. Both equations are now simultaneously written as
\begin{multline}\label{eqn:eqn motion with theta}
\ddot{\Phi}_t(x) = - F_\theta\left(\tilde{\rho}_t(\Phi_t(x))\right)\nabla\tilde{\rho}_t(\Phi_t(x)) - \theta\,(\nabla \Wh*[(F_\theta\circ\tilde{\rho}_t)\mu_t])(\Phi_t(x))\\
- \nabla V\left(\Phi_t(x)\right)  - \eta(\Phi_t(x))\,\dot{\Phi}_t(x) + (K*\mu_t)(\Phi_t(x)).
\end{multline}
Here we use the shorthand notation
\begin{equation}
(\nabla \Wh*[(F_\theta\circ\tilde{\rho}_t)\mu_t])(\xi)=\Int{\Omega_t}{}{\nabla\Wh(\xi-y)F_\theta(\tilde{\rho}_t(y))}{\mu_t(dy)}.
\end{equation}
In \eqref{eqn:eqn motion with theta} we slightly abuse notation, and the equation should be read as follows: whenever $\theta=0$ we disregard the complete term $\theta\,(\nabla \Wh*[(F_\theta\circ\tilde{\rho}_t)\mu_t])(\Phi_t(x))$, irrespective of whether the convolution term is well-defined, bounded etc.\\
\begin{remark}
We emphasize that $F_0$ and $F_1$ are physically different objects in the sense that $F_0$ contains all contributions of $\bar{F}$ to the flow, while $F_1$ only contains part of that influence. Hence, although the notation might suggest so, by setting $\theta=1$ we are not \textit{adding} terms. We use one function $F_\theta$ to facilitate the presentation in the sequel. However, $F_0$ and $F_1$ do have the same physical dimension and e.g.~if $\bar{F}$ is given by $\bar{F}(\rho)\sim\rho^\kappa$ for some $\kappa\in\R\setminus\{0\}$, then both $F_0,F_1\sim \rho^{\kappa-1}$.
\end{remark}
Now we arrive at the central evolution problem we will consider in the rest of this paper. Fix a final time $T>0$. Let $\P(\R^d)$ be the space of probability measures on $\R^d$. Assume that $\mu_0\in\P(\R^d)$ and that there is an $r_0>0$ such that
\begin{equation}\label{eqn: mu0 in ball}
\supp\mu_0\subset B(r_0).
\end{equation}
Let $v_0\in C^1_b(\R^d;\R^d)$ and $\theta\in\{0,1\}$ be fixed. We consider the system
\begin{equation}\label{eqn:system mu0}
\left\{
  \begin{array}{l}
    \ddot{\Phi}_t(x) = - F_\theta\left(\tilde{\rho}_t(\Phi_t(x))\right)\nabla\tilde{\rho}_t(\Phi_t(x)) - \theta\,(\nabla \Wh*[(F_\theta\circ\tilde{\rho}_t)\mu_t])(\Phi_t(x))\\
    \hspace*{5cm}- \nabla V\left(\Phi_t(x)\right)  - \eta(\Phi_t(x))\,\dot{\Phi}_t(x) + (K*\mu_t)(\Phi_t(x));\\
    \tilde{\rho}_t:=\Wh*\mu_t;\\
    \mu_t=\Phi_t\#\mu_0;\\
    \Phi_0(x)=x,\, \dot{\Phi}_0(x)=v_0(x),
  \end{array}
\right.
\end{equation}
for all $x\in\supp\mu_0$ and all $t\in(0,T]$. We remark that this condition implies the one with \eqref{eqn: eqn motion ACB}: that equation is required to hold for almost all $t\in[0,T]$ and for $\mu_0$-almost every $x$.\\
\begin{remark}
We might have taken $\bar{K}*(\Wh*\mu_t)$ for some $\bar{K}$, instead of $K*\mu_t$, to comply with the pressure term (i.e.~the one involving $F_\theta$) that only depends on the \textit{regularized} density $\tilde{\rho}_t$. We prefer the shorter form $K*\mu_t$. This choice can be made without loss of generality if we take $K=\bar{K}*\Wh$.
\end{remark}
\begin{remark}
It is not \textit{a priori} clear whether the term $K*\mu_t$ is a conservative or a nonconservative force density, hence whether it should be part of $q$ or be related to $e$. Assume there is a $\bar{K}$ such that $K(\xi)=-\bar{K}'(|\xi|)\xi/|\xi|$. Then both ways give the same equations of motion. Indeed, if we include the energy density $\frac12\bar{K}*\mu_t$ in $e$ instead of including $K*\mu_t$ in $q$, we also obtain \eqref{eqn:eqn motion with theta}.
\end{remark}

\section{Convergence}\label{sec: conv}
In this section we introduce some preliminary notions, and summarize the required assumptions together with the convergence result (Theorem~\ref{thm: main thm}). The theorem provides a general result, of which the convergence of SPH schemes is a special case; see Corollary~\ref{cor: particle system}. The proof of the theorem in given in Section \ref{sec: proof}.

\subsection{Preliminaries}
Fix a constant integer $d\in\Np$.\\
\begin{definition}[Push-forward]\label{def:push-forward}
The \textit{push-forward} of a probability measure $\mu\in\P(\R^d)$ by a mapping $\map{\Phi}{\R^d}{\R^d}$, notation $\Phi\#\mu$, is defined by
\begin{equation}
(\Phi\#\mu)(B):=\mu(\Phi^{-1}(B))
\end{equation}
for all measurable $B\subset \R^d$. Equivalently, we can define $\Phi\#\mu$ as the push-forward of $\mu$ by $\Phi$ if
\begin{equation}
\Int{\R^d}{}{f(x)}{(\Phi\#\mu)(dx)}=\Int{\R^d}{}{f(\Phi(x))}{\mu(dx)}
\end{equation}
for all measurable, bounded functions $f$ on $\R^d$.
\end{definition}
\begin{definition}[Joint representation]\label{def:coupling}
A \textit{joint representation} of two measures $\mu_1,\mu_2\in\P(\R^d)$ is a measure $\pi$ on $\R^d\times \R^d$ such that
\begin{equation}
\pi(A\times \R^d)=\mu_1(A),\,\,\,\text{and}\,\,\, \pi(\R^d\times B)=\mu_2(B),
\end{equation}
for all $A$ and $B$ in the Borel $\sigma$-algebra of $\R^d$. We denote by $\Pi(\mu_1,\mu_2)$ the set of all joint representations of $\mu_1$ and $\mu_2$. Joint representations are also called \textit{couplings}.
\end{definition}
A useful property of a joint representation $\pi\in\Pi(\mu_1,\mu_2)$ is that for each $i=1,2$
\begin{equation}
\label{eq: marginals}
\Int{\R^d\times\R^d}{}{f(x_i)}{\pi(dx_1,dx_2)}=\Int{\R^d}{}{f(x)}{\mu_i(dx)}
\end{equation}
for all measurable, bounded functions $f$ on $\R^d$. In fact, this is an alternative definition.\\

\begin{definition}[Wasserstein distance]\label{def:Wass}
The \textit{Wasserstein distance} between two probability measures $\mu_1,\mu_2\in\P(\R^d)$ is defined as
\begin{equation}
\mathcal{W}(\mu_1,\mu_2):=\inf_{\pi\in\Pi(\mu_1,\mu_2)}\Int{\R^d\times\R^d}{}{|x-y|}{\pi(dx,dy)}.
\end{equation}
\end{definition}
Note that, to be more precise, we should call this the \textit{$1$-Wasserstein distance}, as a special case of the \textit{$p$-Wasserstein distance} for which the cost function $|x-y|^p$ is used in the integral. The $1$-Wasserstein distance is usually written as $W_1$, but we will stick to $\mathcal{W}$ to avoid confusion with the smoothing function $\Wh$. The particular choice $p=1$ is made because it is compatible with the Lipschitz properties of the functions and the motion mapping that we use. This is what Section \ref{sec: conv} hinges on. For an exposition on the Wasserstein distance and the related concept of \textit{optimal transport}, we refer to \cite{Vil03} and \cite{Vil09}.


\subsection{Assumptions}
Throughout the paper, we assume the following:\\
\begin{assumption}\label{ass:V eta K}
The functions $V$, $\eta$ and $K$ satisfy $V\in C^2_b(\R^d;\R)$, $\eta\in C^1_b(\R^d;\Rp)$ and $K\in C^1_b(\R^d;\R^d)$.
\end{assumption}
\begin{remark}
Note in particular that the above assumption implies that $\nabla V$ and $K$ are Lipschitz continuous. We denote their Lipschitz constants by $|\nabla V|_L$ and $|K|_L$, respectively.
\end{remark}
For $F_\theta$ and $\Wh$ we have requirements that depend on the value of $\theta$. Recall that 
\begin{equation*}
\Rp:=(0,\infty)
\end{equation*}
and define 
\begin{equation*}
\Rp_0:=[0,\infty).
\end{equation*}\\
\begin{assumption}\label{ass:Wh}
The function $\Wh\in C^2_b(\R^d;\Rp_0)$ is even and satisfies $\int_{\R^d}\Wh(x)\,dx=1$.
\end{assumption}

\begin{assumption}[The case $\theta=0$]\label{ass: F W theta=0}
We require that $F_0\in C^1(\Rp;\R)$. Moreover, we assume that there is a constant $M_1>0$ such that for all $\mu\in\P(\R^d)$
\begin{equation}\label{eqn:ass uniform bound Fnabla rho}
\sup_{x\in\supp\mu}\left|F_0\left((\Wh*\mu)(x)\right)\nabla(\Wh*\mu)(x) \right| \leqslant M_1.
\end{equation}
\end{assumption}
If $\theta=0$, we define $M_2,M_3>0$ such that
\begin{align}
\displaystyle\sup_{u\in U_{T,\Wh}}\left|F_0(u) \right|&\leqslant M_2,\,\,\,\text{and}\\
\displaystyle\sup_{u\in U_{T,\Wh}}\left|F'_0(u) \right|&\leqslant M_3,
\end{align}
where
\begin{align}
\displaystyle U_{T,\Wh} &:=\left\{u\in\Rp:\left(\inf_{B(2r(T))}\Wh\right)\leqslant u \leqslant\|\Wh\|_\infty\right\},\,\,\,\text{and}\\
r(T)&:= r_0 + T\|v_0\|_\infty + \frac12\,T^2\,(\|\nabla V\|_\infty+M_1+\|K\|_\infty),
\end{align}
cf.~\eqref{eqn:def r}. Under Assumption \ref{ass: F W theta=0}, $F_0$ may have singularities at the origin, but \textit{only} if $\Wh$ is strictly positive everywhere in $B(2r(T))$. Such $F_0$ and $\Wh$ are used in \cite{DiLisio}; see also Section \ref{sect:Discussion}.\\
\\
If $\theta=1$ we need the following assumption:\\
\begin{assumption}[The case $\theta=1$]\label{ass: F W theta=1}
We assume that $F_1\in C^1(\Rp_0;\R)$. 
\end{assumption}
For $\theta=1$, let $M_2,M_3>0$ be such that
\begin{align}
\sup_{u\in\left[0,\|\Wh\|_\infty\right]}\left|F_1(u)\right|&\leqslant M_2,\,\,\,\text{and}\\
\displaystyle\sup_{u\in\left[0,\|\Wh\|_\infty\right]}\left|F'_1(u) \right|&\leqslant M_3.
\end{align}
and define $M_1:=2\, M_2\,\|\nabla\Wh\|_\infty$.\\
\\
In both cases $\theta=0$ and $\theta=1$, we use the same letters for the constants, to ease notation in the sequel.\\
\begin{remark}
The upper bound in \eqref{eqn:ass uniform bound Fnabla rho} is needed to get an \textit{a priori} bound on the propagation speed in Lemma \ref{lem:bound motion map}. Consequently, we can restrict ourselves to measures with bounded support afterwards; cf.~Corollary \ref{cor: bdd support solution}. To achieve Lemma \ref{lem:bound motion map} if $\theta=1$, we need Assumption \ref{ass: F W theta=1}, which does not allow for singularities in $F_1$ around zero.\\
We demonstrate now why a weaker assumption for $F_1$, resembling \eqref{eqn:ass uniform bound Fnabla rho} is not feasible. Assume that $F_1(\rho):=\rho^\alpha$ with $\alpha\in(-1,0)$. This is the case also considered in \cite{DiLisio}. To bound the first term on the right-hand side of \eqref{eqn:eqn motion with theta}, in \cite{DiLisio} it is assumed that $|\nabla\Wh(\xi)|\leqslant c |\Wh(\xi)|^{-\alpha}$ for some $c>0$. We would need an estimate on
\begin{equation}\label{eqn:ass needed if theta=1}
\sup_{x\in\supp\mu}\left|(\nabla \Wh*[(F_1\circ(\Wh*\mu))\cdot\mu])(x)\right|.
\end{equation}
Let $\Wh$ be strictly positive everywhere. Since $\Wh\in L^1(\R^d)$, $\lim_{\xi\to\infty}\Wh(\xi)=0$. Let $\Wh$ satisfy the aforementioned condition $|\nabla\Wh(\xi)|\leqslant c |\Wh(\xi)|^{-\alpha}$. Then also $\lim_{\xi\to\infty}|\nabla\Wh(\xi)|=0$. Under these (not very strict) conditions one can show that \eqref{eqn:ass needed if theta=1} is unbounded; to see this, use e.g.~the sequence of measures $(\mu^\kappa)_{\kappa\in\Np}$ defined by $\mu^\kappa:=(\delta_{-\kappa\underline{e}_1}+\delta_{\kappa\underline{e}_1}+\delta_{(\kappa+1)\underline{e}_1})/3$, where $\underline{e}_1$ is the first unit vector in $\R^d$. Note in particular that \eqref{eqn:ass needed if theta=1} is unbounded for a Gaussian $\Wh$. The Gaussian however is one of the standard choices for $\Wh$ that we do want to allow for.
\end{remark}

\subsection{Main convergence result}
Let $\{\mu_0^n\}_{n\in\N}\subset\P(\R^d)$, and assume that
\begin{equation}
\supp\mu_0^n\subset B(r_0)\,\,\,\text{ for all }n\in\N,
\end{equation}
where $r_0>0$ is the same constant as in \eqref{eqn: mu0 in ball}. For each $n\in\N$ we associate to the measure $\mu_0^n$ a system of equations analogous to \eqref{eqn:system mu0}:
\begin{equation}\label{eqn:system mu0n}
\left\{
  \begin{array}{l}
    \ddot{\Phi}^n_t(x) = - F_\theta\left(\tilde{\rho}^n_t(\Phi^n_t(x))\right)\nabla\tilde{\rho}^n_t(\Phi^n_t(x)) - \theta\,(\nabla \Wh*[(F_\theta\circ\tilde{\rho}^n_t)\mu^n_t])(\Phi^n_t(x))\\
    \hspace*{4.5cm}- \nabla V\left(\Phi^n_t(x)\right)  - \eta(\Phi^n_t(x))\,\dot{\Phi}^n_t(x) + (K*\mu^n_t)(\Phi^n_t(x));\\
    \tilde{\rho}^n_t:=\Wh*\mu^n_t;\\
    \mu^n_t=\Phi^n_t\#\mu^n_0;\\
    \Phi^n_0(x)=x,\, \dot{\Phi}^n_0(x)=v_0(x),
  \end{array}
\right.
\end{equation}
for all $x\in\supp\mu^n_0$ and all $t\in[0,T]$. Note that the only difference with \eqref{eqn:system mu0} lies in the initial distribution $\mu_0^n$ versus $\mu_0$; the initial velocity $v_0$ is the same.\\
\\
For any $r>0$, define $\P_{r}(\R^d):=\{\mu\in\P(\R^d):\supp\mu\subset B(r)\}$. We also define $\mathcal{A}$ as the space of all functions from $\supp\mu_0$ to $C^2([0,T];\R^d)$.\\
\\
The main result of the present paper is the following.\\
\begin{theorem}
\label{thm: main thm}
Assume that $v_0\in C^1_b(\R^d;\R^d)$, and that Assumptions \ref{ass:V eta K} and \ref{ass:Wh} hold. Let moreover (depending on the value of $\theta$) Assumption \ref{ass: F W theta=0} or \ref{ass: F W theta=1} be satisfied, and take the sequence $\{\mu^n_0\}\subset\P_{r_0}(\R^d)$ such that
\begin{equation}
\label{eqn: initial condition}
\mathcal{W}(\mu_0^n,\mu_0)\stackrel{n\to\infty}{\longrightarrow}0,
\end{equation}
for some $\mu_0\in\P_{r_0}(\R^d)$. Then:
\begin{enumerate}
  \item there is a unique pair $(\mu,\Phi)\in C([0,T];\P_{r(T)}(\R^d))\times \mathcal{A}$ that satisfies \eqref{eqn:system mu0};\label{thm:part1}
  \item if, for all $n\in\N$, the pair $(\mu^n,\Phi^n)\in C([0,T];\P_{r(T)}(\R^d))\times \mathcal{A}$ is a solution of \eqref{eqn:system mu0n}, then
        \begin{equation}
            \sup_{t\in[0,T]}\mathcal{W}(\mu^n_t,\mu_t)\stackrel{n\to\infty}{\longrightarrow}0.
        \end{equation}\label{thm:part2}
\end{enumerate}
\end{theorem}
As a corollary, we obtain the following convergence of the SPH scheme with $n$ particles.\\
\begin{corollary}\label{cor: particle system}
Fix $\theta\in\{0,1\}$. For each $n\in\Np$, let $\bar{\mu}^n_0:=\sum_{j=1}^n m_j\delta_{x_{j,0}}\in\P_{r_0}(\R^d)$ for some $\{m_j\}_{j=1}^n\subset\Rp$ such that $\sum_{j=1}^nm_j=1$, and for some $\{x_{j,0}\}_{j=1}^n\subset B(r_0)$. Assume that $\mathcal{W}(\bar{\mu}^n_0,\mu_0)\stackrel{n\to\infty}{\longrightarrow}0$ for some $\mu_0\in\P_{r_0}(\R^d)$. Then the discrete measure $\bar{\mu}^n_t=\sum_{k=1}^n m_k\delta_{\Phi_t(x_{k,0})}$ associated to the particle scheme defined for each $k\in\{1,\ldots,n\}$ by:
\begin{multline}\label{eqn: cor particle scheme}
\ddot{\Phi}_t(x_{k,0}) = -\sum_{i=1}^n m_i\,\nabla \Wh(\Phi_t(x_{k,0})-\Phi_t(x_{i,0}))\,\left[F_\theta\left(\tilde{\rho}_t(\Phi_t(x_{k,0}))\right)+\theta F_\theta\left(\tilde{\rho}_t(\Phi_t(x_{i,0}))\right)\right] \\
- \nabla V\left(\Phi_t(x_{k,0})\right) - \eta(\Phi_t(x_{k,0}))\,\dot{\Phi}_t(x_{k,0}) + (K*\bar{\mu}^n_t)(\Phi_t(x_{k,0})),
\end{multline}
converges to the solution $\mu_t$ of \eqref{eqn:system mu0} in the following sense:
\begin{equation}
\sup_{t\in [0,T]}\mathcal{W}(\bar{\mu}^n_t,\mu_t)\stackrel{n\to\infty}{\longrightarrow}0.\label{eqn: conv result Wass to zero}
\end{equation}
\end{corollary}

\subsection{Proof of the main convergence theorem}\label{sec: proof}


Before proving the main result, Theorem~\ref{thm: main thm}, we need two auxiliary lemmas concerning the properties of the motion mapping $\Phi_t$. The first lemma is an upper estimate for $\Phi_t$.\\
\begin{lemma}\label{lem:bound motion map}
Let Assumptions \ref{ass:V eta K}, \ref{ass:Wh} and \ref{ass: F W theta=0} or \ref{ass: F W theta=1} (depending on the value of $\theta$) be satisfied. Then for any \textbf{given} $\mu\in C([0,T];\P(\R^d))$ the mapping $\Phi_t$ in \eqref{eqn:eqn motion with theta}, completed with $\tilde{\rho}_t:=\Wh*\mu_t$, $\Phi_0(x)=x$ and $\dot{\Phi}_0(x)=v_0(x)$, satisfies
\begin{equation}\label{eqn:boundMap}
|\Phi_t(x)| \leqslant |x| + t\|v_0\|_\infty + \frac12\,t^2\,(M_1+\|\nabla V\|_\infty+\|K\|_\infty),
\end{equation}
for all $x\in\supp\mu_{0}$ and all $t\in[0,T]$.
\end{lemma}
\begin{proof}
For $\mu$ fixed, and for each $x\in\supp\mu_{0}$, the ODE \eqref{eqn:eqn motion with theta} is well-posed on $[0,T]$, given the assumptions on $V$, $\eta$, $F_\theta$, $\Wh$ and $K$, and the fact that $\mu$ is continuous in time. The well-posedness follows from the Picard-Lindel\"{o}f Theorem; further details are omitted.\\
Using an integrating factor $H(t):=\exp\left(\Int{0}{t}{\eta(\Phi_\tau(x))}{d\tau}\right)$, we deduce from \eqref{eqn:eqn motion with theta} that
\begin{align*}
|\Phi_t(x)| \leqslant& |\Phi_0(x)|+|v_0(x)\,t| +\left|\int_0^t \dfrac{1}{H(s)} \int_0^sH(r)\Big(\ddot{\Phi}_r(x)+\eta(\Phi_r(x))\,\dot{\Phi}_r(x)\Big)\,dr\,ds\right|\\
\leqslant& |x|+t\,\|v_0\|_\infty +\int_0^t \int_0^s\dfrac{H(r)}{H(s)}\left|\ddot{\Phi}_r(x)+\eta(\Phi_r(x))\,\dot{\Phi}_r(x)\right|\,dr\,ds.
\end{align*}
Since $\eta$ is a positive function and hence $0\leqs H(r)/H(s)\leqs1$ in the inner integral, it follows that
\begin{multline}
|\Phi_t(x)|\leqslant |x|+t\,\|v_0\|_\infty + \int_0^t\int_0^s\bigg|-\nabla V(\Phi_r(x))+(\K*\mu_r)(\Phi_r(x))\\
\hspace{2cm}-F_\theta\left((\Wh*\mu_r)(\Phi_r(x))\right)\nabla(\Wh*\mu_r)(\Phi_r(x))\\
-\theta\,(\nabla \Wh*[(F_\theta\circ(\Wh*\mu_r))\mu_r])(\Phi_r(x))) \bigg|\,dr\,ds. \label{eqn:est bound Phi eqn motion filled in}
\end{multline}
In the case $\theta=0$, the following estimate holds due to Assumption \ref{ass: F W theta=0}:
\begin{multline}\label{eqn:est F term theta=0}
\left|F_\theta\left((\Wh*\mu_r)(\Phi_r(x))\right)\nabla(\Wh*\mu_r)(\Phi_r(x))+\theta\,(\nabla \Wh*[(F_\theta\circ(\Wh*\mu_r))\mu_r])(\Phi_r(x)))\right|\\
= \left|F_0\left((\Wh*\mu_r)(\Phi_r(x))\right)\nabla(\Wh*\mu_r)(\Phi_r(x))\right| \leqslant M_1.
\end{multline}
Note that for any $\mu\in\P(\R^d)$ it holds that $\|\Wh*\mu\|_\infty \leqslant \|\Wh\|_\infty$. Hence, in the case $\theta=1$:
\begin{multline}\label{eqn:est F term theta=1}
\left|F_\theta\left((\Wh*\mu_r)(\Phi_r(x))\right)\nabla(\Wh*\mu_r)(\Phi_r(x))+\theta\,(\nabla \Wh*[(F_\theta\circ(\Wh*\mu_r))\mu_r])(\Phi_r(x)))\right|\\
\leqslant M_2\,\|\nabla\Wh\|_\infty + \|\nabla\Wh\|_\infty\,M_2 = M_1,
\end{multline}
where the bounds of Assumption \ref{ass: F W theta=1} are used.\\
\\
A combination of \eqref{eqn:est bound Phi eqn motion filled in}, \eqref{eqn:est F term theta=0} and \eqref{eqn:est F term theta=1} yields that for each $\theta\in\{0,1\}$
\begin{equation}
\nonumber |\Phi_t(x)|\leqslant |x|+t\,\|v_0\|_\infty + \int_0^t\int_0^s\left(\|\nabla V\|_\infty+\|K\|_\infty+M_1\right)\,dr\,ds
\end{equation}
holds for all $x\in\supp\mu_{0}$ and $t\in[0,T]$, from which the statement of the lemma follows.
\end{proof}
\begin{corollary}\label{cor: bdd support solution}
Let $\mu_0\in\P_{r_0}(\R^d)$, and let Assumptions \ref{ass:V eta K}, \ref{ass:Wh} and \ref{ass: F W theta=0} or \ref{ass: F W theta=1} (depending on the value of $\theta$) be satisfied. Then any solution of \eqref{eqn:system mu0} must satisfy
\begin{equation}
\supp\mu_t \subset B(r(t)),
\end{equation}
for each $t\in[0,T]$, where
\begin{equation}\label{eqn:def r}
r(t):= r_0 + t\|v_0\|_\infty + \frac12\,t^2\,(M_1+\|\nabla V\|_\infty+\|K\|_\infty).
\end{equation}
\end{corollary}
The next lemma provides a Lipschitz-like estimate on $\Phi_t$.
\begin{lemma}\label{lem:Lipsch like}
Let $\nu^1,\nu^2\in C([0,T];\P_{r(T)}(\R^d))$ be given. Consider the motion mappings corresponding to $\nu^i$ ($i=1,2$):
\begin{align}\label{eqn:motion map nu}
\ddot{\Phi}^{\nu^i}_t(\xi) =& - F_\theta\left((\Wh*\nu^i_t)(\Phi^{\nu^i}_t(\xi))\right)\nabla(\Wh*\nu^i_t)(\Phi^{\nu^i}_t(\xi)) \\
\nonumber&- \theta\,(\nabla \Wh*[(F_\theta\circ(\Wh*\nu^i_t))\nu^i_t])(\Phi^{\nu^i}_t(\xi))\\
\nonumber&-\nabla V(\Phi^{\nu^i}_t(\xi))-\eta(\Phi^{\nu^i}_t(\xi))\dot{\Phi}^{\nu^i}_t(\xi)+(K*\nu^i_t)(\Phi^{\nu^i}_t(\xi))
\end{align}
for all $\xi\in\supp\nu^i_0$ and all $t\in[0,T]$, with initial conditions $\Phi^{\nu^i}_0(\xi)=\xi$, $\dot{\Phi}^{\nu^i}_0(\xi)=v_0(\xi)$.
Then, for all $t\in[0,T]$, $x\in\supp\nu^1_0$ and $y\in\supp\nu^2_0$, it holds that
\begin{multline}
|\Phi^{\nu^1}_t(x)-\Phi^{\nu^2}_t(y)|\leqslant (1+t\,\|\eta\|_\infty)\,|x-y|+t\,|v_0(x)-v_0(y)|+\\
+ \Int{0}{t}{\left[M_4\,(t-s)+\|\eta\|_\infty\right]\,|\Phi^{\nu^1}_s(x)-\Phi^{\nu^2}_s(y)|}{ds}+ M_5\,\Int{0}{t}{(t-s)\mathcal{W}(\nu_s^1,\nu_s^2)}{ds},\label{eqn:finalestimateGeneral}
\end{multline}
where
\begin{align}
M_4&:=|\nabla V|_L+(1+\theta)M_2 \,\|D^2\Wh\|_\infty+ M_3\,\|\nabla\Wh\|^2_\infty+|K|_L,\,\,\,\text{and}\\
M_5&:=(1+\theta)M_2 \,\|D^2\Wh\|_\infty+ (1+\theta)M_3\,\|\nabla\Wh\|^2_\infty+|K|_L.
\end{align}
\end{lemma}
\begin{proof}
Note that, by the Fubini's theorem, for any integrable function $f$, we have
\begin{equation}
\label{eqn: Fubini thm}
\int_0^t\int_0^r f(s)\,ds\,dr=\int_0^t\int_s^t f(s)\,dr\,ds=\int_0^t (t-s)f(s)\,ds.
\end{equation}
Integration of \eqref{eqn:motion map nu} in time together with \eqref{eqn: Fubini thm} yields that
\begin{align}
\label{eqn:estimate1} |\Phi^{\nu^1}_t(x)-\Phi^{\nu^2}_t(y)| 
\leqslant& |x-y| + t\,|v_0(x)-v_0(y)| \\
\nonumber &+ \Int{0}{t}{(t-s)|\nabla V(\Phi^{\nu^1}_s(x))-\nabla V(\Phi^{\nu^2}_s(y))|}{ds}\\
\nonumber &+ \left|\Int{0}{t}{\Int{0}{r}{\eta(\Phi^{\nu^1}_s(x))\dot{\Phi}^{\nu^1}_s(x)-\eta(\Phi^{\nu^2}_s(y))\dot{\Phi}^{\nu^2}_s(y)}{ds}}{dr}\right|\\
\nonumber &+ \int_0^t(t-s)\left|F_\theta\left((\Wh*\nu^1_s)(\Phi^{\nu^1}_s(x))\right)\nabla(\Wh*\nu^1_s)(\Phi^{\nu^1}_s(x))\right.\\
\nonumber &\hspace{2cm}\left.- F_\theta\left((\Wh*\nu^2_s)(\Phi^{\nu^2}_s(y))\right)\nabla(\Wh*\nu^2_s)(\Phi^{\nu^2}_s(y))\right|\,ds\\
\nonumber &+ \theta \int_0^t(t-s)\left|(\nabla \Wh*[(F_\theta\circ(\Wh*\nu^1_s))\nu^1_s])(\Phi^{\nu^1}_s(x))\right.\\
\nonumber &\hspace{2cm}\left.- (\nabla \Wh*[(F_\theta\circ(\Wh*\nu^2_s))\nu^2_s])(\Phi^{\nu^2}_s(y))\right|\,ds\\
\nonumber &+ \Int{0}{t}{(t-s)\left| (K*\nu^1_s)(\Phi^{\nu^1}_s(x)) - (K*\nu^2_s)(\Phi^{\nu^2}_s(y)) \right|}{ds}.
\end{align}
Furthermore, we have
\begin{equation}
\Int{0}{t}{(t-s)|\nabla V(\Phi^{\nu^1}_s(x))-\nabla V(\Phi^{\nu^2}_s(y))|}{ds}\leqslant |\nabla V|_L\,\Int{0}{t}{(t-s)|\Phi^{\nu^1}_s(x)-\Phi^{\nu^2}_s(y)|}{ds},\label{eqn:estTerm2}
\end{equation}
and
\begin{multline}
\left|\Int{0}{t}{\Int{0}{r}{\eta(\Phi^{\nu^1}_s(x))\dot{\Phi}^{\nu^1}_s(x)-\eta(\Phi^{\nu^2}_s(y))\dot{\Phi}^{\nu^2}_s(y)}{ds}}{dr}\right| = \left|\Int{0}{t}{\Int{0}{r}{\dfrac{d}{ds}\left[\Int{\Phi^{\nu^2}_s(y)}{\Phi^{\nu^1}_s(x)}{\eta(z)}{dz}\right]}{ds}}{dr}\right|\\
= \left|\Int{0}{t}{\left[\Int{\Phi^{\nu^2}_r(y)}{\Phi^{\nu^1}_r(x)}{\eta(z)}{dz}-\Int{y}{x}{\eta(z)}{dz}\right]}{dr}\right| \leqslant\, \|\eta\|_\infty\Int{0}{t}{|\Phi^{\nu^1}_r(x)-\Phi^{\nu^2}_r(y)|}{dr}+\|\eta\|_\infty\,t\,|x-y|.\label{eqn:estTerm3}
\end{multline}
Regarding the term involving $F_\theta$ on the third and fourth line of \eqref{eqn:estimate1}, we proceed as follows
\begin{align}
\label{eqn:estMappingTriangle}&\left|F_\theta\left((\Wh*\nu^1_s)(\xi_1)\right)\nabla(\Wh*\nu^1_s)(\xi_1)- F_\theta\left((\Wh*\nu^2_s)(\xi_2)\right)\nabla(\Wh*\nu^2_s)(\xi_2)\right| \\
\nonumber \leqslant & \, \left|F_\theta\left((\Wh*\nu^1_s)(\xi_1)\right)\right| \,\left|\nabla(\Wh*\nu^1_s)(\xi_1)-\nabla(\Wh*\nu^1_s)(\xi_2) \right|\\
\nonumber&+ \left|F_\theta\left((\Wh*\nu^1_s)(\xi_1)\right)\right| \,\left|\nabla(\Wh*\nu^1_s)(\xi_2)-\nabla(\Wh*\nu^2_s)(\xi_2) \right|\\
\nonumber&+ \left|F_\theta\left((\Wh*\nu^1_s)(\xi_1)\right)-F_\theta\left((\Wh*\nu^1_s)(\xi_2)\right)\right| \,\left|\nabla(\Wh*\nu^2_s)(\xi_2) \right|\\
\nonumber &+ \left|F_\theta\left((\Wh*\nu^1_s)(\xi_2)\right)-F_\theta\left((\Wh*\nu^2_s)(\xi_2)\right)\right| \,\left|\nabla(\Wh*\nu^2_s)(\xi_2) \right|.
\end{align}
We only consider $\xi_1\in\supp\nu^1_s$ and $\xi_2\in\supp\nu^2_s$. This implies that $\xi_1,\xi_2\in B(r(T))$. For each $i,j\in\{1,2\}$, we have the following estimates:
\begin{equation}
\nonumber(\Wh*\nu^i_s)(\xi_j) \leqslant \|\Wh\|_\infty,
\end{equation}
(since $\nu^i_s$ is a probability measure), and
\begin{equation}
\nonumber (\Wh*\nu^i_s)(\xi_j) \geqslant \inf_{\xi_j,z\in B(r(T))}\Wh(\xi_j-z) = \inf_{B(2r(T))}\Wh.
\end{equation}
Thus we get $(\Wh*\nu^i_s)(\xi_j)\in U_{T,\Wh}$. We proceed with the estimation of \eqref{eqn:estMappingTriangle}:
\begin{multline}
\left|F_\theta\left((\Wh*\nu^1_s)(\xi_1)\right)\nabla(\Wh*\nu^1_s)(\xi_1)- F_\theta\left((\Wh*\nu^2_s)(\xi_2)\right)\nabla(\Wh*\nu^2_s)(\xi_2)\right| \\
\leqslant  \,M_2 \,\|D^2\Wh\|_\infty\,\left|\xi_1-\xi_2\right| + M_2 \,\left|\nabla(\Wh*\nu^1_s)(\xi_2)-\nabla(\Wh*\nu^2_s)(\xi_2) \right|\hspace{1cm}\\
+ M_3\,\|\nabla\Wh\|^2_\infty\,\left|\xi_1-\xi_2\right| + M_3\,\|\nabla\Wh\|_\infty\,\left|(\Wh*\nu^1_s)(\xi_2)-(\Wh*\nu^2_s)(\xi_2)\right|,\label{eqn:estMappTriangle followup}
\end{multline}
where we used that $\|\nabla(\Wh*\nu^2_s)\|_\infty\leqslant\|\nabla\Wh\|_\infty$, $\|D^2(\Wh*\nu^2_s)\|_\infty\leqslant\|D^2\Wh\|_\infty$ and the fact that $|\psi|_L=\|\nabla\psi\|_\infty$ for any differentiable function $\psi$. Note that:
\begin{align*}
\left|(\Wh*\nu^1_s)(\xi_2)-(\Wh*\nu^2_s)(\xi_2)\right|&=\left|\Int{}{}{\Wh(\xi_2-z)}{\nu^1_s(dz)}-\Int{}{}{\Wh(\xi_2-w)}{\nu^2_s(dw)}  \right|\nonumber
\\&=\left| \Int{}{}{(\Wh(\xi_2-z)- \Wh(\xi_2-w) )}{\tilde{\pi}_s(dz,dw)} \right|\nonumber
\\&
\leqslant\Int{}{}{\left|\Wh(\xi_2-z)- \Wh(\xi_2-w) \right|}{\tilde{\pi}_s(dz,dw)} \nonumber
\\&\leqslant \|\nabla\Wh\|_\infty\Int{}{}{\left|z-w\right|}{\tilde{\pi}_s(dz,dw)},
\end{align*}
where $\tilde{\pi}_s\in\Pi(\nu^1_s,\nu^2_s)$ is arbitrary and the second equality follows from \eqref{eq: marginals}. By minimizing over all couplings in $\Pi(\nu^1_s,\nu^2_s)$, we obtain
\begin{equation}
\label{eqn:estNu1-Nu2}
\left|(\Wh*\nu^1_s)(\xi_2)-(\Wh*\nu^2_s)(\xi_2)\right|\leqslant \|\nabla\Wh\|_\infty \mathcal{W}(\nu_s^1,\nu_s^2).
\end{equation}
We stress that the bound \eqref{eqn:estNu1-Nu2} is independent of the choice of $\xi_1, \xi_2$. Analogously, we have
\begin{align}
\left|\nabla(\Wh*\nu^1_s)(\xi_2)-\nabla(\Wh*\nu^2_s)(\xi_2)\right| \leqslant \|D^2\Wh\|_\infty \mathcal{W}(\nu_s^1,\nu_s^2).\label{eqn: est nabla coupling}
\end{align}
It follows that
\begin{multline}
\left|F_\theta\left((\Wh*\nu^1_s)(\xi_1)\right)\nabla(\Wh*\nu^1_s)(\xi_1)- F_\theta\left((\Wh*\nu^2_s)(\xi_2)\right)\nabla(\Wh*\nu^2_s)(\xi_2)\right| \\
\leqslant \left(\,M_2 \,\|D^2\Wh\|_\infty+ M_3\,\|\nabla\Wh\|^2_\infty\right)\,\left(\left|\xi_1-\xi_2\right|+\mathcal{W}(\nu_s^1,\nu_s^2)\right).\label{eqn:finalEstTerm4}
\end{multline}
If $\theta=1$, similar estimates as in the first term on the right-hand side of \eqref{eqn:estMappTriangle followup}, and as in \eqref{eqn:estNu1-Nu2} and \eqref{eqn: est nabla coupling} yield
\begin{multline}
\left|(\nabla \Wh*[(F_\theta\circ(\Wh*\nu^1_s))\nu^1_s])(\xi_1) - (\nabla \Wh*[(F_\theta\circ(\Wh*\nu^2_s))\nu^2_s])(\xi_2)\right| \\
\leqslant \, M_2 \,\|D^2\Wh\|_\infty\,\left|\xi_1-\xi_2\right|+ \left(M_2 \,\|D^2\Wh\|_\infty + M_3\,\|\nabla\Wh\|^2_\infty\right)\,\mathcal{W}(\nu_s^1,\nu_s^2).\label{eqn:estMappingTriangle theta}
\end{multline}


The last term in \eqref{eqn:estimate1} we treat as follows:
\begin{align}
\label{eqn:estTerm5} \left| (K*\nu^1_s)(\xi_1) - (K*\nu^2_s)(\xi_2) \right|\leqslant & \left| (K*\nu^1_s)(\xi_1) - (K*\nu^1_s)(\xi_2) \right| \\
\nonumber &+ \left| (K*\nu^1_s)(\xi_2) - (K*\nu^2_s)(\xi_2) \right|\\
\nonumber \leqslant & |K|_L\left(|\xi_1-\xi_2| + \mathcal{W}(\nu_s^1,\nu_s^2)\right).
\end{align}
The estimate of the second term is obtained like in \eqref{eqn:estNu1-Nu2}.\\
\\
We combine \eqref{eqn:estimate1}, \eqref{eqn:estTerm2}, \eqref{eqn:estTerm3}, \eqref{eqn:finalEstTerm4}, \eqref{eqn:estMappingTriangle theta} and \eqref{eqn:estTerm5} to get
\begin{multline}
\nonumber |\Phi^{\nu^1}_t(x)-\Phi^{\nu^2}_t(y)|\leqslant\, (1+t\,\|\eta\|_\infty)\,|x-y|+ t\,|v_0(x)-v_0(y)|\\
+ \Int{0}{t}{\left[M_4\,(t-s)+\|\eta\|_\infty\right]\,|\Phi^{\nu^1}_s(x)-\Phi^{\nu^2}_s(y)|}{ds} + M_5\,\Int{0}{t}{(t-s)\mathcal{W}(\nu_s^1,\nu_s^2)}{ds},
\end{multline}
with $M_4$ and $M_5$ as defined in the statement of the lemma.
\end{proof}
Now we have all ingredients to prove the main theorem, Theorem \ref{thm: main thm}.\\
\\
\textbf{Proof of Part \ref{thm:part1} of Theorem \ref{thm: main thm}.}\\If $M_5=0$ then the well-posedness of \eqref{eqn:system mu0} is straightforward. In that case, $F_\theta\circ(\Wh*\mu_t)=0$ on $\supp\mu_t$ for all $t$ and moreover $K$ must be constant, so the first equation in \eqref{eqn:system mu0} is independent of $\mu_t$. The Picard-Lindel\"{o}f Theorem guarantees, for each $x\in\supp\mu_0$, existence and uniqueness of the motion mapping (as mentioned before). The solution $(\mu_t)_{0\leqslant t\leqslant T}$ is uniquely defined by the push-forward $\mu_t=\Phi_t\#\mu_0$.\\
\\
If $M_5\neq0$, the well-posedness proof is based on a fixed-point argument (Banach's Fixed Point Theorem). Let $T>0$ be fixed. Choose $N\in\Np$ large enough, such that $T^*:=T/N$ satisfies
\begin{equation}\label{eqn:kappa}
\kappa_{T^*}:=\frac12\,(T^*)^2\,M_5\,\exp\left(\|\eta\|_\infty\, T^* + \frac12\,M_4\,(T^*)^2  \right)<1.
\end{equation}
Let $j\in\{1,\ldots,N\}$ be fixed. Suppose that $\mu_0^{(j)}\in\P(\R^d)$ and $v^{(j)}_0\in C^1_b(\R^d;\R^d)$ are given. Consider a mapping $\mathcal{F}^{(j)}\colon\nu\mapsto\mu:=\mathcal{F}^{(j)}(\nu)$ from
\begin{equation}\label{eqn: def Cj}
\mathcal{C}_j:=\left\{\nu\in C([0,T^*];\P_{r(jT^*)}(\R^d)): \nu|_{t=0}=\mu^{(j)}_0\right\}
\end{equation}
to itself, defined by
\begin{equation}
\mu_t=[\Phi^{(j)}_t]^\nu\#\mu^{(j)}_0,  \hbox{ for all $t\in[0,T^*]$,}
\end{equation}
where the motion mapping $[\Phi^{(j)}]^\nu: \supp\mu^{(j)}_0 \rightarrow C^2([0,T^*];\R^d)$ is the solution to the following ODE
\begin{equation}\label{eqn:mapping}
\left\{
  \begin{array}{l}
    [\ddot{\Phi}^{(j)}_t]^{\nu}(x) = -F_\theta\left(\tilde{\rho}_t([\Phi^{(j)}_t]^\nu(x))\right)\nabla\tilde{\rho}_t([\Phi^{(j)}]^\nu_t(x))\\
      \hspace*{2.5cm}- \theta\,(\nabla \Wh*[(F_\theta\circ\tilde{\rho}_t)\nu_t])([\Phi^{(j)}_t]^\nu(x))\\
    \hspace*{2.5cm}-\nabla V\left([\Phi^{(j)}_t]^\nu(x)\right)-\eta\left([\Phi^{(j)}_t]^\nu(x)\right)\,[\dot{\Phi}^{(j)}_t]^\nu(x) + (K*\nu_t)([\Phi^{(j)}_t]^\nu(x));\\
    \tilde{\rho}_t:=\Wh*\nu_t;\\
    \Phi^\nu_0(x)=x,\, \dot{\Phi}^\nu_0(x)=v^{(j)}_0(x).
  \end{array}
\right.
\end{equation}
The space $\mathcal{C}_j$ is complete for arbitrary $j\in\{1,\ldots,N\}$ due to Theorem \ref{thm:completenessAppendix} in Appendix \ref{app:completeness}.
Note that a fixed point $\mu^{(j)}$ of this mapping together with the corresponding motion mapping $\Phi^{(j)}$ is a solution of \eqref{eqn:system mu0} on $[0,T^*]$ with initial data $\mu^{(j)}_0$ and $v_0^{(j)}$. We create a hierarchy of the mappings $\mathcal{F}^{(j)}$ for $j=1,\ldots,N$ by defining $\mu^{(j+1)}_0:=\mu^{(j)}_{T^*}$, $\mu^{(1)}_0:=\mu_0$, $v^{(j+1)}_0:=\dot{\Phi}^{(j)}_{T^*}$ and $v^{(1)}_0:=v_0$. Such definition only makes sense if mapping $j$ actually has a unique fixed point and thus $\mu^{(j)}_{T^*}$ and $\dot{\Phi}^{(j)}_{T^*}$ are well-defined. Moreover, we are aware of the fact that we have only defined $v^{(j+1)}_0$ on the support of $\mu^{(j+1)}_0$. This is however sufficient. If we insist, we might just define it to be zero outside. In view of the to be constructed hierarchy, $\supp\mu^{(j)}_0\subset B(r((j-1)T^*))$ should be satisfied for each $j$.\\
\\
For any $\nu\in\mathcal{C}_j$ the image $\mu=\mathcal{F}^{(j)}(\nu)$ exists, and actually is an element of $\mathcal{C}_j$. Well-posedness of the motion mapping (for given $\nu$ and for each $x\in\supp\mu^{(j)}_0$) follows from Picard-Lindel\"{o}f (see before) and guarantees the existence and uniqueness of $\mu$.\\
The support of the image measure, $\supp\mu_t$, is contained in a ball of radius
\begin{equation}
r(jT^*) = r_0 + jT^*\,\|v_0\|_\infty + \frac12(jT^*)^2\,(M_1+\|\nabla V\|_\infty+\|K\|_\infty).
\end{equation}
This is easily checked by use of \eqref{eqn:boundMap} and a recursive relation involving $\|[\dot{\Phi}^{(j)}_{T^*}]^\nu\|_\infty=\|v^{(j+1)}_0\|_\infty$ for each $j\in\{1,\ldots,N-1\}$. Thus, the image $\mu$ of our mapping $\mathcal{F}^{(j)}$ is an element of $\mathcal{C}_j$.\\
Consider two measures $\nu^1,\nu^2\in\mathcal{C}_j$ and their corresponding images $\mu^1:=\mathcal{F}^{(j)}(\nu^1)$ and $\mu^2:=\mathcal{F}^{(j)}(\nu^2)$. Let $\pi_0\in\Pi(\mu_0^{(j)},\mu_0^{(j)})$ be arbitrary. For an arbitrary fixed $t\in[0,T^*]$ define $\pi_t\in\Pi(\mu^1_t,\mu^2_t)$ by
\begin{equation}
\pi_t:=\left([\Phi^{(j)}_t]^{\nu^1},[\Phi^{(j)}_t]^{\nu^2}\right)\#\pi_0.
\end{equation}
Note that this $\pi_t$ is indeed a joint representation of $\mu^1_t$ and $\mu^2_t$ for each $t$. We drop the dependence on $j$ of $\pi_0,\mu^1,\mu^2$ and $\pi_t$ since no ambiguity appears. By definition of the push-forward and of the Wasserstein distance (see Definitions \ref{def:push-forward} and \ref{def:Wass}), we have
\begin{equation}\label{eqn:WassAgainstMapping}
\mathcal{W}(\mu^1_t,\mu^2_t) \leqslant \Int{}{}{|z-w|}{\pi_t(dz,dw)} = \Int{}{}{\left|[\Phi^{(j)}_t]^{\nu^1}(x)-[\Phi^{(j)}_t]^{\nu^2}(y)\right|}{\pi_0(dx,dy)}
\end{equation}
holds for each $t\in[0,T^*]$. Applied to \eqref{eqn:finalestimateGeneral}, a version of Gronwall's Lemma
~yields that for each $x,y\in\supp\mu^{(j)}_0$
\begin{multline}
\left|[\Phi^{(j)}_t]^{\nu^1}(x)-[\Phi^{(j)}_t]^{\nu^2}(y)\right| \leqslant\, \bigg[(1+t\,\|\eta\|_\infty)\,|x-y| + t\,|v^{(j)}_0(x)-v^{(j)}_0(y)|\\
\hspace{1cm}+  M_5\,\Int{0}{t}{(t-s)\mathcal{W}(\nu_s^1,\nu_s^2)}{ds} \bigg]\,\exp\left(\|\eta\|_\infty\, t + \frac12\,M_4\,t^2  \right).\label{eqn:estGronwall}
\end{multline}
We remark that Gronwall's Lemma may be applied because the term $|v^{(j)}_0(x)-v^{(j)}_0(y)|$ is bounded and $s\mapsto (t-s)\mathcal{W}(\nu_s^1,\nu_s^2)$ is bounded and continuous. The former can be shown by using estimates similar to those in the proof of Lemma \ref{lem:bound motion map}. The boundedness of the latter is trivial, \linebreak $(t-s)\mathcal{W}(\nu_s^1,\nu_s^2)\leq T^*\sup_{s\in[0,T^*]}\mathcal{W}(\nu_s^1,\nu_s^2)$; while the continuity of $s\mapsto \mathcal{W}(\nu_s^1,\nu_s^2)$ follows from the triangle inequality. Indeed, since \begin{equation*}
|\mathcal{W}(\nu_s^1,\nu_s^2)-\mathcal{W}(\nu_{s_0}^1,\nu_{s_0}^2)|\leq \mathcal{W}(\nu_s^1,\nu_{s_0}^1)+\mathcal{W}(\nu_{s_0}^2,\nu_{s}^2),
\end{equation*}
it implies that $\lim_{s\to s_0}\mathcal{W}(\nu_s^1,\nu_s^2)=\mathcal{W}(\nu_{s_0}^1,\nu_{s_0}^2)$ if $\mathcal{W}(\nu_s^1,\nu_{s_0}^1)\to 0$ and $\mathcal{W}(\nu_s^2,\nu_{s_0}^2)\to 0$.\\
\\
Now we combine \eqref{eqn:WassAgainstMapping} and \eqref{eqn:estGronwall}, and obtain
\begin{multline}\label{eqn: Wass against pi0}
\mathcal{W}(\mu^1_t,\mu^2_t) \leqslant \bigg[(1+t\,\|\eta\|_\infty)\Int{}{}{|x-y|}{\pi_0(dx,dy)} + t\,\Int{}{}{|v^{(j)}_0(x)-v^{(j)}_0(y)|}{\pi_0(dx,dy)}\\
\hspace{1cm}+M_5\,\Int{0}{t}{(t-s)\mathcal{W}(\nu_s^1,\nu_s^2)}{ds}\bigg]\,\exp\left(\|\eta\|_\infty\, t + \frac12\,M_4\,t^2  \right).
\end{multline}
The integral with respect to $\pi_0(dx,dy)$ disappeared for the third term inside the square brackets, since this term is independent of $x$ and $y$, and moreover $\Int{}{}{}{\pi_0(dx,dy)}=1$. Now we take
\begin{equation*}
\pi_0 := (I\otimes I)\# \mu_0^{(j)},
\end{equation*}
which is the measure concentrated on the diagonal $x=y$ with marginals both $\mu_0^{(j)}$. With some abuse of notation it can also be written as
\begin{equation*}
\pi_0(dx,dy) := \delta(x-y)\mu_0^{(j)}(dy).
\end{equation*}
For this choice of $\pi_0$, we have that
\begin{align*}
\Int{}{}{|x-y|}{\pi_0(dx,dy)}=0,\,\,\text{and}\\
\Int{}{}{|v^{(j)}_0(x)-v^{(j)}_0(y)|}{\pi_0(dx,dy)}=0.
\end{align*}
Therefore, only the third term in square brackets on the right-hand side of \eqref{eqn: Wass against pi0} remains. Since
\begin{equation*}
\int_0^t(t-s)\mathcal{W}(\nu_s^1,\nu_s^2)\,ds\leq \sup_{s\in[0,t]}\mathcal{W}(\nu^1_s,\nu^2_s)\int_0^t(t-s)\,ds=\frac{1}{2}t^2\sup_{s\in[0,t]}\mathcal{W}(\nu^1_s,\nu^2_s),
\end{equation*}
we obtain
\begin{equation}\nonumber
\mathcal{W}(\mu^1_t,\mu^2_t) \leqslant \frac12 t^2\,M_5\,\exp\left(\|\eta\|_\infty\, t + \frac12\,M_4\,t^2  \right)\,\sup_{s\in[0,t]}\mathcal{W}(\nu^1_s,\nu^2_s).
\end{equation}
Finally, we take the supremum over $t\in[0,T^*]$:
\begin{align}
\sup_{t\in[0,T^*]}\nonumber \mathcal{W}(\mu^1_t,\mu^2_t) \leqslant& \frac12\,(T^*)^2\,M_5\,\exp\left(\|\eta\|_\infty\, T^* + \frac12\,M_4\,(T^*)^2  \right)\,\sup_{t\in[0,T^*]}\mathcal{W}(\nu^1_t,\nu^2_t).
\end{align}
By the specific choice of $T^*$, $\mathcal{F}^{(j)}$ is a contraction mapping for each $j$, since
\begin{align}
\sup_{t\in[0,T^*]}\nonumber \mathcal{W}(\mu^1_t,\mu^2_t) \leqslant& \,\kappa_{T^*}\,\sup_{t\in[0,T^*]}\mathcal{W}(\nu^1_t,\nu^2_t),
\end{align}
where $\kappa_{T^*}<1$ by assumption; cf.~\eqref{eqn:kappa}. As mentioned before, the space $\mathcal{C}_j$ is complete for each $j$ due to Theorem \ref{thm:completenessAppendix} in Appendix \ref{app:completeness}. Banach's Fixed Point Theorem then guarantees the existence of a unique fixed point of $\mathcal{F}^{(j)}$ for each $j$.\\
Having the construction of $(\mu^{(j)},\Phi^{(j)})$ for $j=1,\ldots,N$, we define a couple $(\mu,\Phi)$ of a measure and a motion mapping as follows
\begin{equation}
(\mu_t,\Phi_t):=(\mu^{(j)}_{t-(j-1)T^*},\Phi^{(j)}_{t-(j-1)T^*}),\quad \text{if}~~ t\in ((j-1)T^*,jT^*],
\end{equation}
for $j\in\{1,\ldots,N\}$.\\
\noindent By our construction $(\mu,\Phi)\in C([0,T];\P_{r(T)}(\R^d))\times \mathcal{A}$ and it uniquely satisfies \eqref{eqn:system mu0} with initial data $\mu_0$ and $v_0$.\\
\\
%
%
\textbf{Proof of Part \ref{thm:part2} of Theorem \ref{thm: main thm}.}\\Note that Part \ref{thm:part1} implies that for each initial measure $\mu_0$ and $\mu_0^n$ (for each $n\in\N$) there is a corresponding unique solution $(\mu,\Phi)$, $(\mu^n,\Phi^n)$, respectively. Fix $n\in\N$ and let $\pi_0\in\Pi(\mu_0^n,\mu_0)$ be arbitrary. We use \eqref{eqn:finalestimateGeneral}, taking $\nu^1=\mu^n$ and $\nu^2=\mu$. Thus $\Phi^{\nu^1}=\Phi^n$ and $\Phi^{\nu^2}=\Phi$. First of all, we estimate
\begin{equation}
|v_0(x)-v_0(y)|\leqslant \|\nabla v_0\|_\infty\,|x-y|,
\end{equation}
for all $x\in\supp\mu_0^n$ and all $y\in\supp\mu_0$. This is possible\footnote{Note that this estimate was not possible in \eqref{eqn:estGronwall}, since $v^{(j)}_0$ is part of the solution and only defined on $\supp \mu_0^{(j)}$. In general, $\nabla v^{(j)}_0$ might not even be defined.}, since $v_0\in C^1_b(\R^d;\R^d)$ is given, is defined on the whole of $\R^d$ and has bounded derivative. Using this Lipschitz estimate and integrating \eqref{eqn:finalestimateGeneral} against $\pi_0(dx,dy)$, we obtain
\begin{multline}
\Int{}{}{|\Phi^n_t(x)-\Phi_t(y)|}{\pi_0(dx,dy)}\leqslant\, (1+t\,\left(\|\nabla v_0\|_\infty+\|\eta\|_\infty\right))\,\Int{}{}{|x-y|}{\pi_0(dx,dy)}\\
+ \Int{0}{t}{\left[M_4\,(t-s)+\|\eta\|_\infty\right]\,\Int{}{}{|\Phi^n_s(x)-\Phi_s(y)|}{\pi_0(dx,dy)}}{ds} + M_5\,\Int{0}{t}{(t-s) \mathcal{W}(\mu^n_s,\mu_s)}{ds},\label{eqn:integratedEstimatebeforeGronwall}
\end{multline}
where we used that the last term is independent of $x$ and $y$, and the fact that $\pi_0$ is a probability measure on $\R^d\times\R^d$. If we define $\tilde{\pi}_s\in\Pi(\mu^n_s,\mu_s)$ as
\begin{equation}
\tilde{\pi}_s:=(\Phi^n_s,\Phi_s)\#\pi_0
\end{equation}
for each $s\in[0,T]$, then we have, analogously to \eqref{eqn:WassAgainstMapping}, the following:
\begin{equation}
\mathcal{W}(\mu^n_s,\mu_s) \leqslant \Int{}{}{|z-w|}{\tilde{\pi}_s(dz,dw)} = \Int{}{}{\left|\Phi^n_s(x)-\Phi_s(y)\right|}{\pi_0(dx,dy)}.\label{eqn:WassAgainstMapping2}
\end{equation}
We substitute this estimate for $\mathcal{W}(\mu^n_s,\mu_s)$ in the right-hand side of \eqref{eqn:integratedEstimatebeforeGronwall} and apply Gronwall's Lemma to obtain
\begin{multline}
\Int{}{}{|\Phi^n_t(x)-\Phi_t(y)|}{\pi_0(dx,dy)}\\
\leqslant \left[(1+t\,\left(\|\nabla v_0\|_\infty+\|\eta\|_\infty\right))\,\Int{}{}{|x-y|}{\pi_0(dx,dy)}\right]\exp\left(\|\eta\|_\infty\, t + \frac12\left(M_4+M_5\right) t^2  \right).\label{eqn:estBeforeWass}
\end{multline}
Together, \eqref{eqn:WassAgainstMapping2} and \eqref{eqn:estBeforeWass} yield
\begin{multline}\nonumber
\mathcal{W}(\mu^n_t,\mu_t)\\ \leqslant \left[(1+t\,\left(\|\nabla v_0\|_\infty+\|\eta\|_\infty\right))\,\Int{}{}{|x-y|}{\pi_0(dx,dy)}\right]\,\exp\left(\|\eta\|_\infty\, t + \frac12\left(M_4+M_5\right) t^2  \right).
\end{multline}
We take the infimum over $\pi_0\in\Pi(\mu^n_0,\mu_0)$ on the right-hand side:
\begin{equation}
\mathcal{W}(\mu^n_t,\mu_t) \leqslant (1+t\,\left(\|\nabla v_0\|_\infty+\|\eta\|_\infty\right))\,\exp\left(\|\eta\|_\infty\, t + \frac12\left(M_4+M_5\right) t^2  \right)\,\mathcal{W}(\mu^n_0,\mu_0).
\end{equation}
Finally, we take the supremum over $t\in[0,T]$ on both sides of the inequality and obtain
\begin{equation}
\nonumber\sup_{t\in[0,T]}\mathcal{W}(\mu^n_t,\mu_t) \leqslant (1+T\,\left(\|\nabla v_0\|_\infty+\|\eta\|_\infty\right))\,\exp\left(\|\eta\|_\infty\, T + \frac12\left(M_4+M_5\right) T^2  \right)\,\mathcal{W}(\mu^n_0,\mu_0).
\end{equation}
Hence $\mathcal{W}(\mu^n_0,\mu_0)\stackrel{n\to\infty}{\longrightarrow}0$ implies $\sup_{t\in[0,T]}\mathcal{W}(\mu^n_t,\mu_t)\stackrel{n\to\infty}{\longrightarrow}0$. This finishes the proof. \endproof 


\subsection{Discussion on Assumptions \ref{ass: F W theta=0} and \ref{ass: F W theta=1}, and the condition \eqref{eqn: initial condition}}\label{sect:Discussion}
We comment on the assumptions needed for the main theorem, Theorem \ref{thm: main thm}.\\
\\
{\em Assumptions on $F_\theta$ and $\Wh$}: We remark here that in \cite{DiLisio} only $\theta=0$ is used, and furthermore $\nabla V\equiv0$, $\eta\equiv0$ and $K\equiv0$. All possible $F_0$ and $\Wh$ treated in \cite{DiLisio} satisfy Assumption \ref{ass: F W theta=0}:
\begin{enumerate}
  \item $F_0(u)=u^\alpha$, for $\alpha\geqslant0$, satisfies the assumptions for all choices of $\Wh\in C^2_b(\R^d;\Rp_0)$;
  \item $F_0(u)=u^\alpha$, for $-1<\alpha<0$, satisfies the assumptions if $\Wh$ is an element of $C^2_b(\R^d;\Rp)$ and satisfies the extra condition $|\nabla\Wh(x)|\leqslant c\,|\Wh(x)|^{-\alpha}$ for all $x$, for some constant $c>0$.
\end{enumerate}
We remark that the class of admissible pairs $(F_0,\Wh)$ covered by Assumption \ref{ass: F W theta=0} is more general than in \cite{DiLisio}, where only $F_0$ of the form $F_0(u)=u^\alpha$ is treated. For instance, in our work any $F_0\in C^1_b(\Rp;\Rp)$ (bounded and with bounded derivative) is allowed in combination with an arbitrary $\Wh\in C^2_b(\R^d;\Rp_0)$.\\
\\
{\em Assumption \eqref{eqn: initial condition} on convergence of initial data}: Given the initial probability measure $\mu_0$ supported in the ball $B(r_0)$, we demonstrate here two ways of constructing an approximating sequence of measures $(\mu_0^n)_{n\in\Np}$.\\
\\
The first way of constructing $\mu_0^n$ is deterministic and has been used in~\cite{Bolley}. For simplicity of presentation, we assume $d=1$ and $\supp\mu_0\subset[0,1]$. For each $n\in\Np$, define
\begin{equation}\label{eq: construction init}
\mu_0^n:=\sum_{i=1}^{n}m_i\delta_{\frac{i}{n}-\frac{1}{2n}},
\end{equation}
where $\displaystyle m_i:=\int_{[\frac{i-1}{n},\frac{i}{n})}\mu_0(dx)$, for each $i=1,\ldots,n-1$, and $\displaystyle m_n:=\int_{[1-\frac{1}{n},1]}\mu_0(dx)$.\\
It follows that $\sum_{i}m_i=\int \mu_0(dx)=1$ and $\mu_0^n\in\P(\R)$.
Define a map $\map{\Psi}{[0,1]}{\{\frac{i}{n}-\frac{1}{2n}:1\leq i\leq n\}}$ by $\Psi(x):=\frac{i}{n}-\frac{1}{2n}$ if $\frac{i-1}{n}\leq x < \frac{i}{n}$ and $\Psi(1):=1-\frac{1}{2n}$. For every measurable and bounded function $f$, defined on $[0,1]$ it holds that
\begin{align*}
\int_{[0,1]}f(x)\mu_0^n(dx)&=\sum_{i=1}^n m_i\,f\left(\frac{i}{n}-\frac{1}{2n}\right)\\
&=\sum_{i=1}^{n-1}\int_{[\frac{i-1}{n},\frac{i}{n})}\mu_0(dx)\, f\left(\frac{i}{n}-\frac{1}{2n}\right)+\int_{[1-\frac{1}{n},1]}\mu_0(dx)\, f\left(1-\frac{1}{2n}\right)\\
&=\int_{[0,1]}f(\Psi(x))\,\mu_0(dx).
\end{align*}
Hence, $\mu_0^n=\Psi\#\mu_0$. Note that $|x-\Psi(x)|\leqslant\frac{1}{2n}$ for every $x\in[0,1]$. Therefore,
\[
\mathcal{W}(\mu_0^n,\mu_0)\leq \int_{[0,1]}|x-\Psi(x)|\,\mu_0(dx) \leq \frac{1}{2n}\int_{[0,1]}\mu_0(dx) =\frac{1}{2n},
\]
where we obtain the first inequality by taking $\pi\in\Pi(\mu_0^n,\mu_0)$ to be $\pi:=(I\otimes\Psi)\#\mu_0$. This implies that $\mathcal{W}(\mu^n_0,\mu_0)\stackrel{n\to\infty}{\longrightarrow}0$.\\
This procedure generalizes to the case $d>1$ (but with more involved notation). Let $\supp\mu_0\subset[0,1]^d$ and let $n\in\{k^d:k\in\Np\}$. Dividing the hypercube $[0,1]^d$ into $n$ equal subcubes, we obtain analogously that the convergence rate is $\mathcal{O}(1/\sqrt[d]{n})$.\\
\\
The second way of constructing $\mu_0^n$ is probabilistic and is based on the law of large numbers as already pointed out in~\cite{DiLisio}. Suppose that the points $X_i, i=1,\ldots,n$ are independent identically distributed random variables with the same distribution $\mu_0\in\P(B(r_0))$. Let $\mu_0^n$ be the empirical measure, defined by
\[
\mu_0^n:=\frac{1}{n}\sum_{i=1}^n\delta_{X_i}.
\]
Note that in fact there is an underlying probability space $\Omega$ and $\map{X_i}{\Omega}{B(r_0)}$. Hence $\mu_0^n$ is, strictly speaking, not a mere probability measure, but a mapping from $\Omega$ to $\P(B(r_0))$; i.e.~$\map{\mu_0^n}{\Omega}{\P(B(r_0))}$. According to \cite{Dudley}, Theorem 11.4.1, the sequence $(\mu_0^n)$ converges \textit{almost surely} to $\mu_0$. This implies that for \textit{almost every} realization $\bar{x}_1,\bar{x}_2,\ldots$ the corresponding sequence of measures $(\bar{\mu}_0^n)\subset\P(B(r_0))$ given by $\bar{\mu}_0^n:=1/n\sum_i\delta_{\bar{x}_i}$, converges in the narrow topology to $\mu_0$:
\[
\int_{B(r_0)}f(x)\mu_0^n(dx)\rightarrow\int_{B(r_0)}f(x)\mu_0(dx),\,\,\,\text{for all $f\in C_b(B(r_0))$.}
\]
The term `almost every realization' refers to the fact that the set (in $\Omega$) on which the narrow convergence does \textit{not} hold, has zero probability (with respect to the probability distribution on $\Omega$). In layman's terms, this means that if we draw a random sample $\bar{x}_1,\bar{x}_2,\ldots$, it is `unlikely' that the corresponding sequence $(\bar{\mu}_0^n)$ does not converge narrowly.\\
\\
Assume that our random sample did yield such narrowly converging sequence $(\bar{\mu}_0^n)$. Since all $\bar{\mu}_0^n$ are probability measures on a bounded domain $B(r_0)$, their first moments are uniformly integrable (i.e.~uniformly in $n$). Thus, Theorem 7.1.5 in \cite{AGS} implies that
\begin{equation}
\mathcal{W}(\bar{\mu}_0^n,\mu_0)\stackrel{n\to\infty}{\longrightarrow}0.
\end{equation}
%

\subsection{Numerical illustration}\label{sec: num}
We illustrate the theoretical convergence result of Theorem \ref{thm: main thm} by two numerical examples. The first one involves only the hydrodynamical force, as described by the first term on the right-hand side of \eqref{eqn: cor particle scheme}. We consider both schemes derived ($\theta = 0$ and $\theta = 1$), in dimension $d=1$ and $d=2$. In the second example only the non-local interaction term and a drag force in \eqref{eqn: cor particle scheme} are present and we take $d=2$. 
First, in problems of bounded domains, it is common to use the differential form of mass conservation equation, thus the time variation of the measure-valued equation for mass in \eqref{eqn:system mu0} is evolved in time along with the momentum equation, via a leapfrog algorithm with a constant time step. The leapfrog algorithm is a second-order symplectic integrator with the property of preserving the momentum of the system. The Gaussian function, defined by
\begin{equation}
\Wh(x) := \dfrac{1}{h \sqrt{\pi}}e^{-|x|^2/h^2},
\end{equation}
for all $x\in\R$, is used for the regularization of the mass measure in the one-dimensional case. For $d=2$, the cubic Wendland function is used, whence for all $x\in\R^2$:
\begin{equation}
\Wh(x) :=
\begin{cases} \frac{1}{8} ( 1 + 3|x|/2h )(2 - |x|/h)^3, & |x| \leq 2h, \\
0, & |x| > 2h.
\end{cases}
\end{equation}
These choices are made to illustrate that we can handle both bounded and unbounded support of $\Wh$.\\
In order for the regularized equations of hydrodynamics to approximate the real physics well, $h$ should be sufficiently small. Let $V_0$ denote a representative volume assigned to each particle based on the initial configuration. In a bounded domain, typically $V_0$ scales as $V_0\sim 1/n$. It is common practice to achieve ``$h$ sufficiently small" by taking $h = \eps\,\sqrt[d]{V_0}$, with parameter $1.2\leqs\eps\leqs1.5$, cf.~\cite{Mon05}. However, the convergence result in Theorem \ref{thm: main thm} holds for $h$ fixed, and the dependence of $h$ on $n$ is not investigated. Numerically, we investigate both cases. That is, we take both $h = 1$ fixed and $h = 1.5\,\sqrt[d]{V_0}$, which hence varies with the number of particles.
\begin{figure}
 	\includegraphics[width=5.8in]{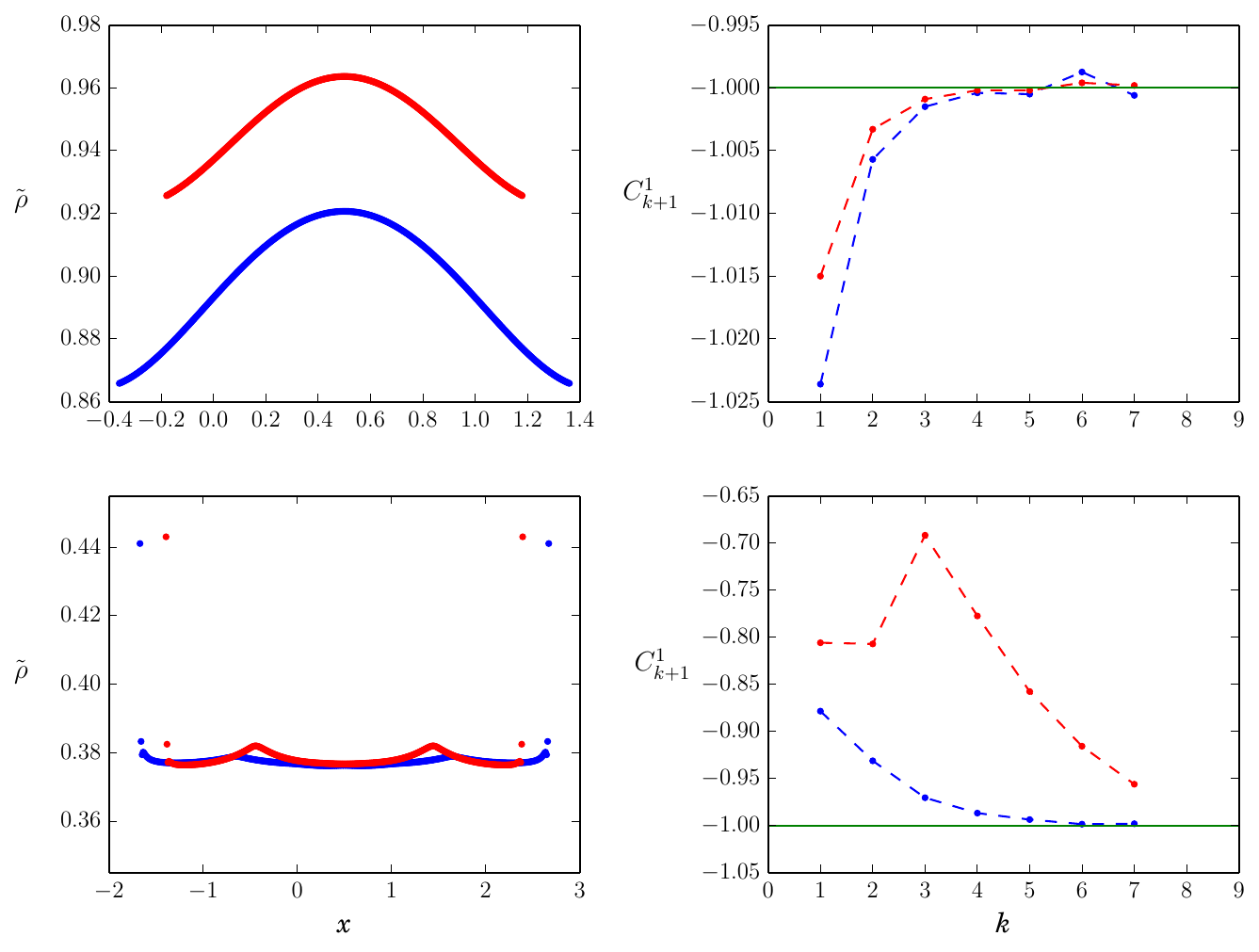}
 	\caption{For $\gamma = 1$, density $\tilde{\rho}$ at final time $T = 1$, with $n= 2^9$ particles and convergence $C^1$. Red and blue plots refer to the schemes for $\theta = 0$ and $\theta = 1$, respectively. Upper plots present results for $h=1$ fixed, while lower plots to variable $h = 1.5\,V_0$.}\label{plot_1d-gamma1}
\end{figure}
We assume that the initial measure $\mu_0$ has a density $\rho_0$ such that $\rho_0(x)=1$ for all $x\in[0,1]^d$ and $\rho_0(x)=0$ otherwise. We construct the measure $\mu^n_0$, corresponding to the $n$-particle approximation, according to \eqref{eq: construction init} or its $d$-dimensional counterpart. Hence, the initial particle configuration is realized for $d=1$ by equipartitioning the initial domain $[0,1]$ into $n$ volumes. For the two-dimensional examples, the initial domain is the square $[0,1]^2\subset\R^2$ and particles are placed in the center of each square incremental volume partition $V_0$. Masses are assigned as $m_i = \rho_0(x_i)\,V_0$ for each $i=1,\ldots,n$. Note that in case of more complicated initial domains, an equipartitioning may be obtained with a centroidal Voronoi tessellation.\\
As argued underneath \eqref{eq: construction init}, the sequence $(\mu^n_0)_{n\in\Np}$ constructed in this way converges to $\mu_0$ and the convergence rate is $\mathcal{O}(1/\sqrt[d]{n})$. Hence, the corresponding solutions $(\mu^n)_{n\in\Np}$ converge at the same rate; see the last lines of the proof of Part \ref{thm:part2} of Theorem \ref{thm: main thm}.\\
\begin{figure}
 	\includegraphics[width=5.8in]{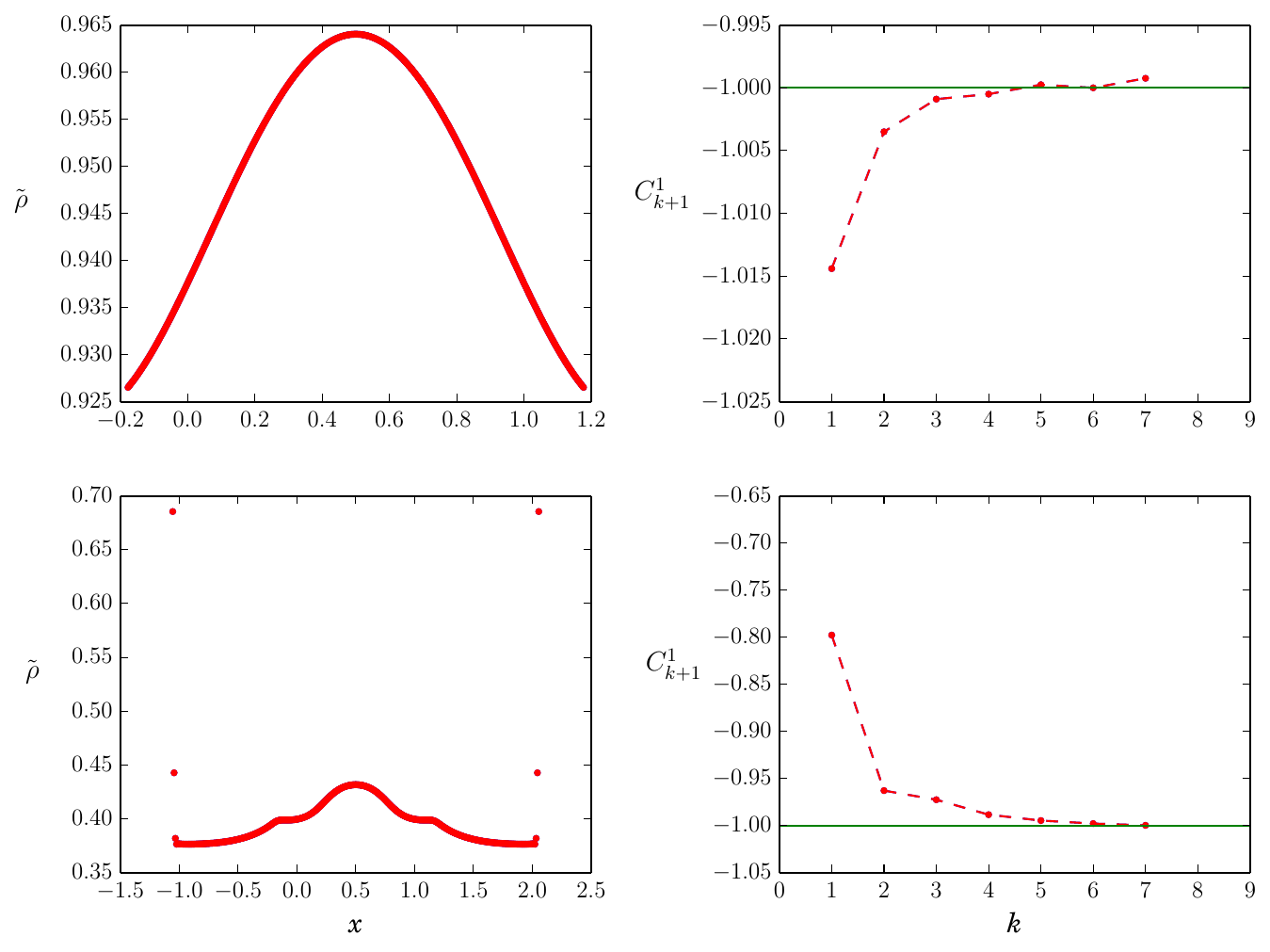}
 	\caption{For $\gamma = 2$, density $\tilde{\rho}$ at final time $T = 1$, with $n= 2^9$ particles and convergence $C^1$. Red and blue plots refer to the schemes for $\theta = 0$ and $\theta = 1$, respectively. Upper plots present results for $h=1$ fixed, while lower plots to variable $h = 1.5\,V_0$.}\label{plot_1d-gamma2}
\end{figure}
\\
The hydrodynamical problem considers the spontaneous expansion of a gas cloud until time $T = 1$, governed by the equation of state $P(\rho) = \mathcal{K}\rho^\gamma$, where $\mathcal{K} = 1$ is a parameter and $\gamma$ the so-called \textit{polytropic exponent}. We recall that $P$ relates to $e$ via $\partial e/\partial \rho=P/\rho^2$. In dimension $d=1$, we examine the cases $\gamma \in \{1, 2, 7\}$, using a constant time step $\Delta t = 10^{-3}$, and the Gaussian function. Note that the case $\gamma = 1$ is not covered by the convergence proof (cf.~Assumption \ref{ass: F W theta=0} and Section \ref{sect:Discussion}). It is a limit case (the proof does hold for any $\gamma>1$) and we include it for generality. We perform the calculations for $n = 2^k$ particles, where $k \in \{1,\ldots, 9\}$, and compute the supremum in time of the Wasserstein distance between subsequent solutions; cf.~\eqref{eqn: conv result Wass to zero}. We compute the Wasserstein distance by solving a linear programming problem based on a formulation in terms of optimal transportation. Due to the high computational cost (for large $k$), we use the following approximation
\begin{equation}\label{eqn: approx sup wass}
\sup_{t\in [0,T]}\mathcal{W}(\mu^{2^k}_t,\mu^{2^{k+1}}_t) \approx \max_{\tau\in I} \mathcal{W}(\mu^{2^k}_\tau,\mu^{2^{k+1}}_\tau) =: W_{k,k+1},
\end{equation}
to reduce the number of evaluations of $\mathcal{W}$. Here,
\begin{equation}
I:= \{jT/(N_r - 1)\,:\,j=0,\ldots,(N_r-1)\}
\end{equation}
and we take $N_r = 10$. It should be noted, however that for the vast majority of the computations, the maximum distance is observed at the final time step. \\
\\
\begin{figure}
 	\includegraphics[width=5.8in]{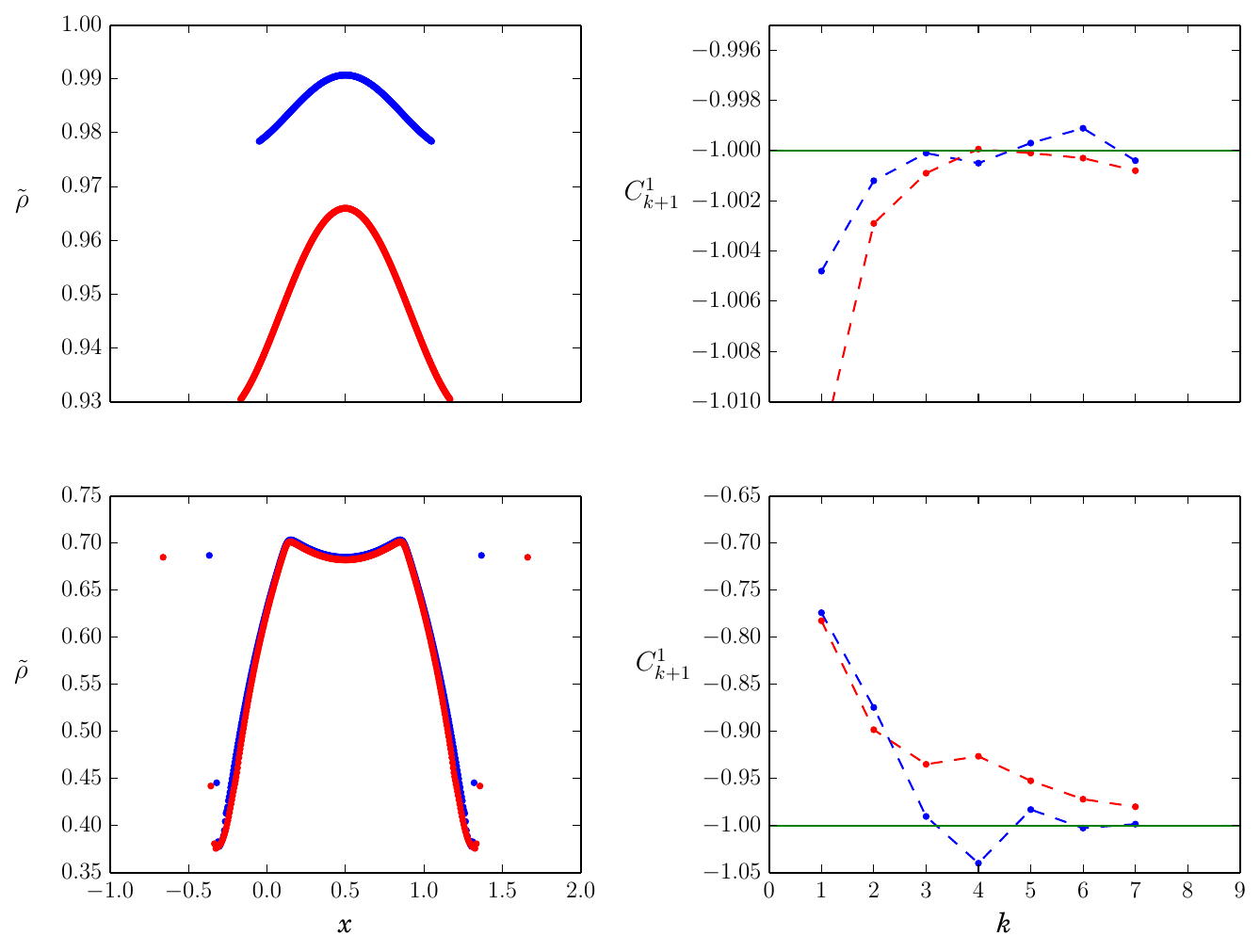}
 	\caption{For $\gamma = 7$, density $\tilde{\rho}$ at final time $T = 1$, with $n= 2^9$ particles and convergence $C^1$. Red and blue plots refer to the schemes for $\theta = 0$ and $\theta = 1$, respectively. Upper plots present results for $h=1$ fixed, while lower plots to variable $h = 1.5\,V_0$.}\label{plot_1d-gamma7}
\end{figure}
The convergence rate for $d=1$ is approximated by
\begin{equation}
C^1_{k+1} := \log_2 \Big |\frac{W_{k+1,k+2}}{W_{k,k+1}}\Big |
\end{equation}
and based on the theoretical prediction that the convergence rate is $\mathcal{O}(n^{-1})$ if $d=1$, we expect that $C^1_{k+1}$ tends to the value $-1$.\\
\\
In Figures \ref{plot_1d-gamma1}-\ref{plot_1d-gamma7}, results are shown for the three different values $\gamma \in \{1, 2, 7\}$ respectively. Red graphs correspond to the scheme for $\theta = 0$ and blue graphs to the scheme $\theta = 1$. In these figures, the upper plots refer to computations using $h = 1$ for all resolutions and the lower plots depict computations with $h$ varying with the number of particles used. The left plots show the result for density $\tilde{\rho}$ at $T = 1$, as obtained with the highest resolution $n = 2^9$. Additionally, the convergence of $C_{k+1}^1$ is plotted in the right plots.\\
\\
There are several points to be mentioned about the plots. First, note that in all figures solutions, for $h=1$ fixed and $h$ varying, do converge to a solution by increasing the number of particles. The convergence is evident by the rate $C_{k+1}^1$ approaching its theoretical value $-1$. Second, although convergent, solutions for $h=1$ fixed and $h$ varying are not the same for the same value of $\gamma$. Third, in Figure \ref{plot_1d-gamma2} where $\gamma = 2$, the solutions obtained with the two schemes coincide. This effect is expected since for this value of $\gamma$ the two schemes are identical. On the other hand, this is not true for the the cases $\gamma = 1$ and $\gamma = 7$. Fourth, interestingly enough, even though the proof only covers cases for $\gamma > 1$, the case $\gamma = 1$ converges. In the same case, it is unclear why a spike is present in the convergence graph for $\theta = 0$ and varying $h$. Fifth, for fixed value of $h$ all cases converge from below sharply towards the theoretical value $C_{k+1}^1 \approx -1$, while for $h$ varying with the resolution they converge from above. 
Finally, for fixed $h=1$, this large value does not permit local effects to appear on the free boundaries of the domain. These effects are exhibited in the cases of varying $h$ as discontinuities of the density profile and therefore seem to be related to problems of applying regularization over small $h$-sized regions in bounded domains.\\
\\
In two spatial dimensions, the hydrodynamic problem examined is the expansion of an initially square gas cloud, until time $T = 1$. In order to show that the results also hold for non-static initial conditions, a rotation described by the initial velocity field $(v_{0,x}, v_{0,y}) = (-y, x)$ is applied. The same equation of state as in the one-dimensional computation is used, with $\gamma \in \{2, 7\} $. The Wendland function and a constant time step of $\Delta t = 10^{-3}$ are employed. Note that we omit the case $\gamma=1$, hence do not need to `mimic' Assumption \ref{ass: F W theta=0}, and do allow for bounded support in $\Wh$.\\
\begin{table}[t]
\caption{
$C^2_{k+1}$ for the two-dimensional hydrodynamic computations}
\centering
\begin{tabular}{c | c c c c c c}
\hline\hline
\\
&$k$ &2&3&4&5& \\ [0.5ex] 
\hline\hline
\\
$\gamma=2$& $\theta=0$ &-0.51&-0.50&-0.50&-0.50\\
$h$ fixed& $\theta=1$ &-0.51&-0.50&-0.50&-0.50\\
[0.5ex]\hline\\
$\gamma=2$& $\theta=0$ &-0.44&-0.47&-0.49&-0.44\\
$h$ varying& $\theta=1$ &-0.44&-0.47&-0.49&-0.44\\
[0.5ex]\hline\hline \\
$\gamma=7$& $\theta=0$ &-0.51&-0.50&-0.50&-0.50\\
$h$ fixed& $\theta=1$ &-0.50&-0.50&-0.50&-0.50\\
[0.5ex]\hline\\
$\gamma=7$& $\theta=0$ &-0.37&-0.45&-0.48&-0.48\\
$h$ varying& $\theta=1$ &-0.41&-0.43&-0.52&-0.51\\
[1ex]
\hline
\end{tabular}
\label{tab: conv2d}
\end{table}
\\
For $d=2$, we approximate the rate of convergence by
\begin{equation}
C^2_{k+1} := \dfrac12 \, \log_2 \Big |\frac{W_{k+1,k+2}}{W_{k,k+1}}\Big |.
\end{equation}
Note that this definition is different from $C^1_{k+1}$, since in $d=1$ we took $n$ of the form $2^k$, while in $d=2$ we have $n=(2^{k})^2$, for $k \in \{1,2,3,4,5,6\}$. Here, the definition of $W_{k,k+1}$ is modified accordingly to approximations by $4^k$ and $4^{k+1}$ particles, respectively. The computational effort for the calculation of the Wasserstein distance makes the investigation of higher $n$ extremely lengthy. In the case $d=2$, theory predicts that the convergence rate is $\mathcal{O}(n^{-1/2})$, whence we expect that $C^2_{k+1}$ tends to the value $-1/2$.
\begin{figure}
 	\includegraphics[width=5.8in]{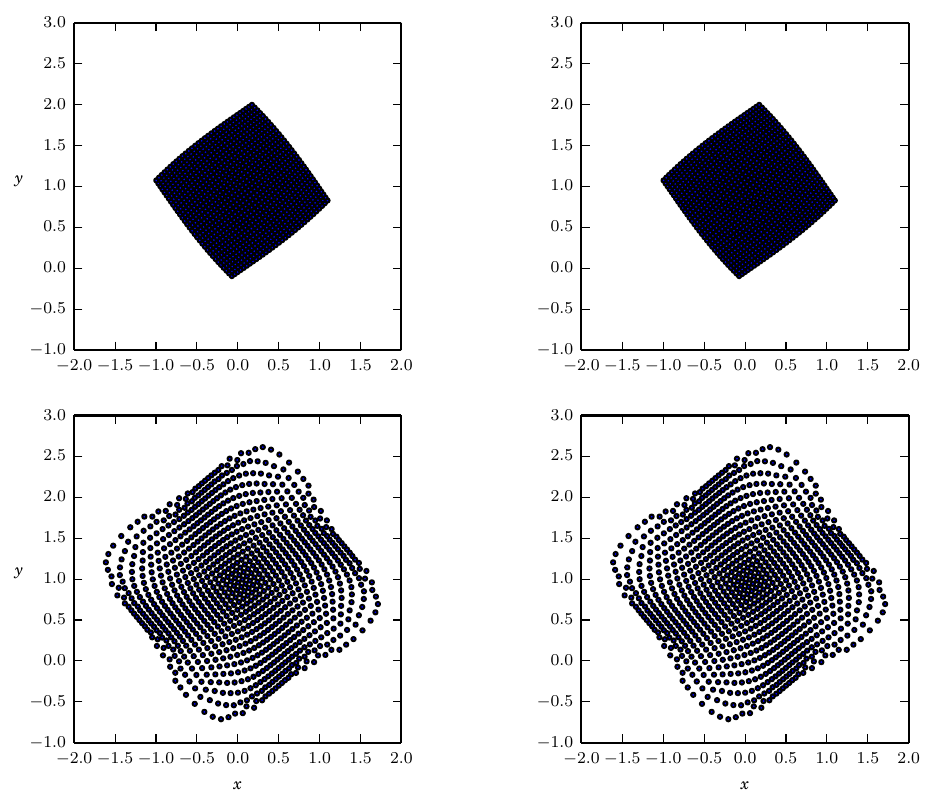}
 	\caption{For the case $\gamma = 2$, particle configurations at final time $T=1$ for the two-dimensional hydrodynamic experiment of a rotating square; on the left-hand side with $\theta = 0$ and on the right-hand side with $\theta =1$. The upper plots present results for $h=1$ and $n= 512$ particles, while the lower row refers to variable $h = 1.5 \sqrt{V_0}$ and $n= 512$ particles. In this case ($\gamma=2$), the schemes for $\theta = 0$ and $\theta =1$ are the same.}\label{plot_2d-gamma2}
\end{figure}
In Table \ref{tab: conv2d}, the convergence rates of the two-dimensional hydrodynamic problems are shown. The theoretical value is indeed approached, but strong oscillations around this value appear in the case $\gamma = 2$ with varying $h$. In Figures \ref{plot_2d-gamma2}-\ref{plot_2d-gamma7}, particle configurations at final time $T=1$ are presented for the cases $\gamma = \{2,7\}$ respectively. The upper plots refer to fixed $h = 1$ independent of the resolution $n$ (a choice in agreement with the convergence proof), while lower plots are obtained by $h$ varying with the number of particles as $h = 1.5\sqrt{V_0}$. For the plots on the left-hand side the scheme with $\theta = 0$ is used, while for the plots on the right-hand side $\theta = 1$ is employed. Similarly to the one-dimensional results, the corresponding solutions for $\gamma = 2$ are identical for the schemes employing $\theta = 0$ or $\theta =1$. On the contrary, they differ for $\gamma = 7$. Finally, it should be mentioned that the instabilities of the density profile on the boundaries of the domain, which were observed in the one-dimensional computations, have now translated into the nonhomogeneous distribution of particles.\\
\\
The second numerical example considers the nonlocal force and the drag term, for which the numerical scheme corresponding to \eqref{eqn: cor particle scheme} does not depend on $\theta$. Moreover, \eqref{eqn: cor particle scheme} does not depend on $\tilde{\rho}_t$, hence $h$ is only relevant if we wish to plot $\tilde{\rho}_t$, and not for the computations themselves.\\
The Wendland function and a constant time step of $\Delta t = 10^{-1}$ are used. For the interactions, we take $K$ such that it is the gradient of the \textit{Morse potential}, see e.g.~\cite{DOrsogna}, with parameters $C_a=2.0  ,\, C_r=1.5  ,\,\ell_a=1.0   , \, \ell_r=2.0$. In fact, we included a short-range regularization around the origin to the potential to enforce the required $C^1_b$-regularity of $K$. A side-effect is that automatically self-interactions are cancelled. Two cases for the drag coefficient are examined: $\eta \equiv 10$ and $\eta \equiv 0.1$, for final time $T=100$. In both these cases, an equilibrium has been reached. Similarly to the hydrodynamical problem, \eqref{eqn: approx sup wass} is used with $n = 2^{2k}$ particles, where $k \in \{1, 2, 3, 4, 5\}$. Particle configurations and convergence rates are plotted in Figure \ref{plot_nonlocal}, with the upper plots referring to $\eta \equiv 0.1$ and the lower plots to $\eta \equiv 10$. The value of the convergence in this case rapidly tends to the theoretically predicted value.
\begin{figure}
 	\includegraphics[width=5.8in]{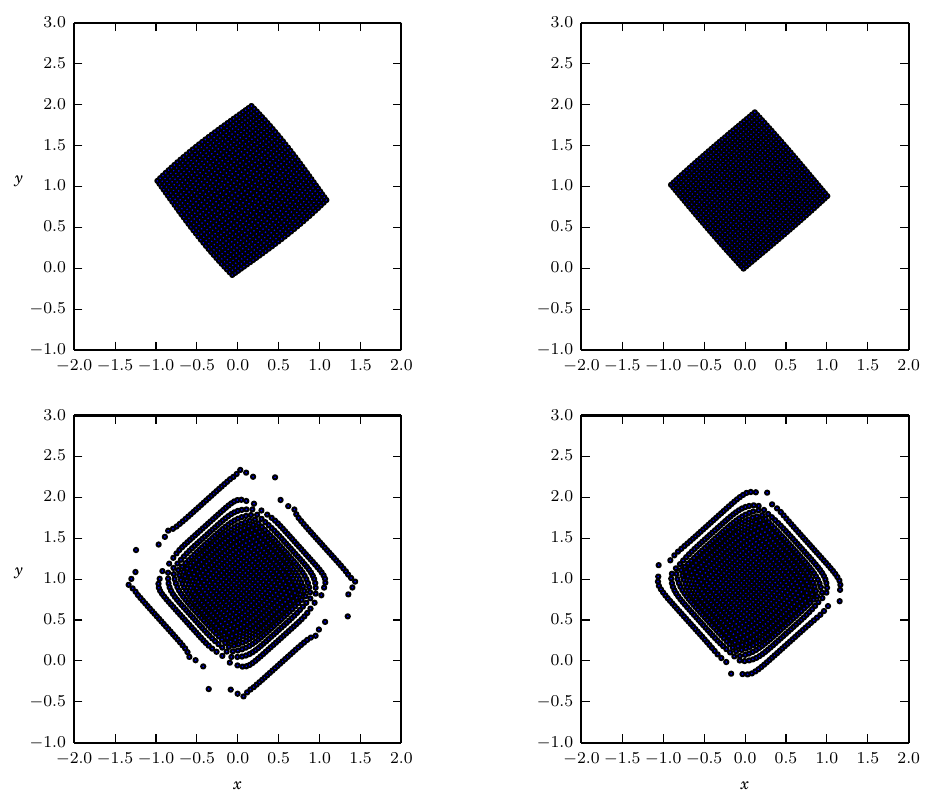}
 	\caption{For the case $\gamma = 7$, particle configurations at final time $T=1$ for the two-dimensional hydrodynamic experiment of a rotating square; on the left-hand side with $\theta = 0$ and on the right-hand side with $\theta =1$. The upper plots present results for $h=1$ and $n= 512$ particles, while the lower row refers to variable $h = 1.5 \sqrt{V_0}$ and $n= 512$ particles.}\label{plot_2d-gamma7}
\end{figure}
\begin{figure}
 	\includegraphics[width=5.8in]{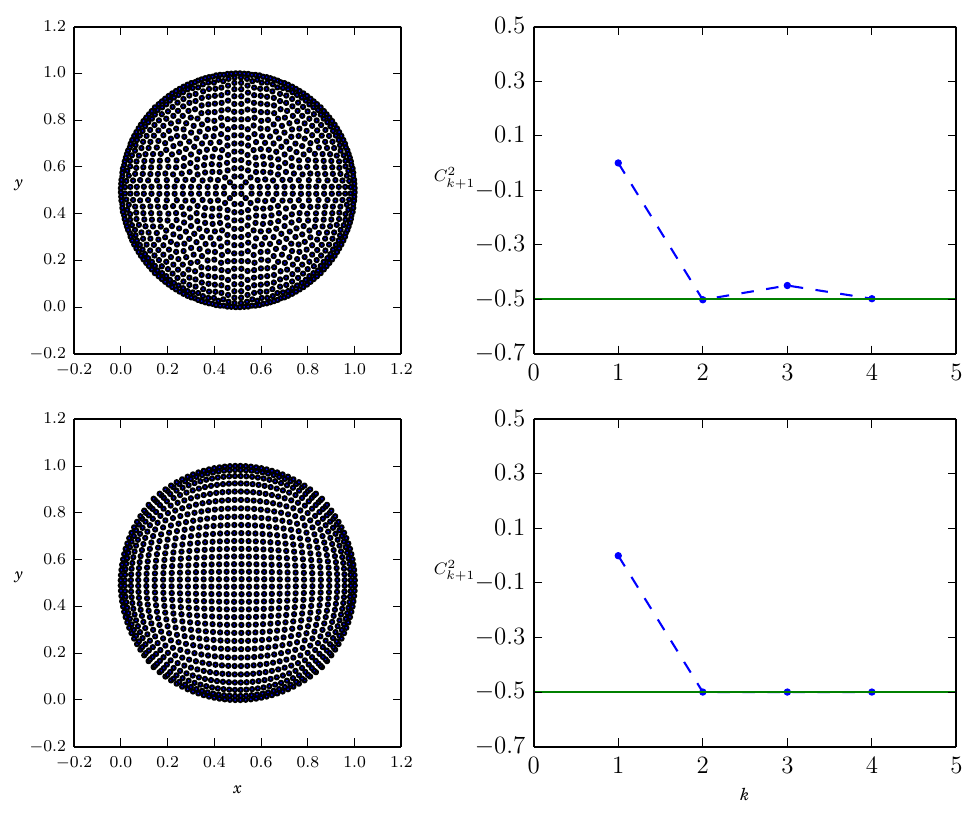}
 	\caption{For the problems involving a nonlocal force term, the particle configurations at final time $T = 100$ (left plots) and the convergence rates (right plots). Drag coefficients $\eta\equiv 0.1$ (upper plot) and $\eta \equiv10$ (lower plot) are used.}\label{plot_nonlocal}
\end{figure}
%
%
\section{Concluding remarks and future directions}\label{sec: concl and future}
Apart from the remarks already made, there are two issues that are important to point out. One could call them shortcomings of our approach, in the sense that these are cases to which our proof of convergence does not apply. The result of Theorem \ref{thm: main thm} does not state:
\begin{itemize}
  \item whether the approximations corresponding to $\theta=0$ and $\theta=1$, respectively, actually converge to the same limit solution. 
      Our computations show that this is certainly not the case for $h=1$ and (although the difference is smaller) neither for varying $h$, except for the trivial case $\gamma=2$ in which the schemes coincide.
  \item whether the limit $n\to\infty$ in any of the two cases $\theta=0$ or $\theta=1$ is actually like the `real physics'. To investigate this, in principle one would need to consider the limit $h\to0$. 
      As said before, this is beyond the scope of the current paper.
\end{itemize}
The latter point refers to a situation in which first the limit $n\to\infty$ is taken and afterwards the limit $h\to0$. A more favourable approach (also from a numerical point of view) would be to have $h$ depend on $n$ in such a way that $h=\mathcal{O}(1/\sqrt[d]{n})$ as $n\to\infty$, and hence $n\to\infty$ and $h\to0$ simultaneously. In Section \ref{sec: num}, we anticipated this ---following what is already typically done in the literature of SPH--- and the numerical results there support the hope that solutions converge in the case of $h$ varying with the number of particles.\\
\\
Nevertheless, our combined theoretical-computational results establish the convergence of the classical and most-used SPH scheme and also show that the corresponding equation of motion is a true discretized version of the equation of motion of a regularized continuous medium.

\section*{Acknowledgements}
We thank Adrian Muntean, Mark Peletier and Fons van de Ven (TU Eindhoven, The Netherlands) for fruitful discussions and useful comments. Until 2015 J.H.M. Evers was a member of the Centre for Analysis, Scientific computing and Applications, and the Institute for Complex Molecular Systems (ICMS) at TU Eindhoven, supported by the Netherlands Organisation for Scientific Research (NWO), Graduate Programme 2010. For I.A. Zisis, this research was carried out under project number M11.4.10412 in the framework of the Research Program of the Materials innovation institute M2i (www.m2i.nl). For B.J. van der Linden it is a research activity of the Laboratory of Industrial Mathematics in Eindhoven LIME bv (www.limebv.nl).

\begin{appendix}

\section{Completeness}\label{app:completeness}
The arguments in this appendix lead to the statement of Theorem \ref{thm:completenessAppendix}. This theorem implies that the space $\mathcal{C}_j$ defined in \eqref{eqn: def Cj} is a complete metric space for every $j\in\{1,\ldots,N\}$. This result is needed to be able to apply Banach's Fixed Point Theorem in the proof of Part \ref{thm:part1} of Theorem \ref{thm: main thm}.\\
\begin{lemma}\label{lem:PBRcomplete}
Fix $R>0$. Then the space $\P(B(R))$ of probability measures on $B(R):=\{x\in\R^d:|x|\leqslant R\}$, endowed with the metric $\mathcal{W}:=W_1$, is a complete metric space.
\end{lemma}
\begin{proof}
Since $B(R)$ is complete, it follows from \cite{AGS}, Proposition 7.1.5, that $(\P_1(B(R)),W_1)$ is complete. Here, $\P_1(B(R))$ is the space of probability measures with bounded first moment. The statement of the lemma follows from the observation that
\begin{equation}
\P_1(B(R))=\P(B(R)).
\end{equation}
The inclusion $\P_1(B(R))\subset\P(B(R))$ is trivial. The other inclusion follows from the fact that the first moment of each $\mu\in\P(B(R))$ is bounded by $R$.
\end{proof}
\begin{lemma}\label{lem:C0TPBRcomplete}
For each $T,R>0$, the space
\begin{equation}
C([0,T];\P(B(R))),
\end{equation}
endowed with the metric
\begin{equation}
\sup_{\tau\in[0,T]}\,\mathcal{W}(\mu_1(\tau),\mu_2(\tau)),
\end{equation}
is complete.
\end{lemma}
\begin{proof}
The proof mainly follows the lines of the proof of Theorem 1.5-5 in \cite{Kreyszig} (which treats real-valued continuous functions).\\
Let $(\mu_n)_{n\in\N}$ denote a Cauchy sequence in $C([0,T];\P(B(R)))$. Fix $\eps>0$. There is a $K$ such that for all $m,n\geqslant K$
\begin{equation}
\sup_{\tau\in[0,T]}\,\mathcal{W}(\mu_m(\tau),\mu_n(\tau))<\eps.
\end{equation}
For any fixed $t\in[0,T]$,
\begin{equation}
\mathcal{W}(\mu_m(t),\mu_n(t))\leqslant\sup_{\tau\in[0,T]}\,\mathcal{W}(\mu_m(\tau),\mu_n(\tau))<\eps
\end{equation}
holds, so $(\mu_n(t))_{n\in\N}$ is a Cauchy sequence in $\P(B(R))$. It follows from Lemma \ref{lem:PBRcomplete} that $\P(B(R))$ is complete and thus $(\mu_n(t))_{n\in\N}$ converges to some $\tilde{\mu}_t\in\P(B(R))$. This pointwise limit exists for every $t\in[0,T]$, and we construct a mapping $\mu$ from $[0,T]$ to $\P(B(R))$ by defining
\begin{equation}
\mu(t):=\tilde{\mu}_t
\end{equation}
for all $t\in[0,T]$.\\
There is an $N$ such that
\begin{equation}\label{eqn:supWmn<eps/2}
\sup_{\tau\in[0,T]}\,\mathcal{W}(\mu_m(\tau),\mu_n(\tau))<\eps/2
\end{equation}
for all $m,n\geqslant N$ (with the same $\eps$ as before!). In particular, for fixed $t\in[0,T]$,
\begin{equation}\label{eqn:WNn<eps/2}
\mathcal{W}(\mu_m(t),\mu_n(t))<\eps/2
\end{equation}
holds for all $m,n\geqslant N$.
Thus, for each fixed $t\in[0,T]$ and for each $m\geqslant N$,
\begin{equation}\label{eqn: Wass mu_m mu small}
\mathcal{W}(\mu_m(t),\mu(t))\leqslant \underbrace{\mathcal{W}(\mu_m(t),\mu_n(t))}_{<\eps/2}+\underbrace{\mathcal{W}(\mu_n(t),\mu(t))}_{<\eps/2}<\eps,
\end{equation}
for sufficiently large $n$. Here we use \eqref{eqn:WNn<eps/2} to estimate the first term on the right-hand side. Due to the fact that $\mu(t)$ is defined as the pointwise limit of $\mu_n(t)$, the second term can be made arbitrarily small by increasing $n$. We conclude from \eqref{eqn: Wass mu_m mu small} that $\mathcal{W}(\mu_m(t),\mu(t))<\eps$ for all $m\geqslant N$. Due to \eqref{eqn:supWmn<eps/2}, this estimate holds with the same $\eps$ and $N$ for all $t\in[0,T]$, whence
\begin{equation}
\sup_{t\in[0,T]}\mathcal{W}(\mu_m(t),\mu(t))\leqslant\eps
\end{equation}
for all $m\geqslant N$, which proves the convergence of $(\mu_m)$ to $\mu$.\\
The limit $\mu:[0,T]\to\P(B(R))$ is continuous since it is the uniform limit of continuous mappings (cf.~\cite{Kosmala} Thm.~8.3.1 for a proof for real-valued functions that can be extended trivially to our situation), hence $(\mu_n)$ converges in $C([0,T];\P(B(R)))$.
\end{proof}
\begin{lemma}[cf.~\cite{Kreyszig} Theorem 1.4-7]\label{lem:closedSubsetComplete}
If $Y$ is a closed subset of a complete metric space $(X,\mathfrak{d})$, then $Y$ is complete.
\end{lemma}
\begin{proof}
Let $(\xi_n)\subset Y$ be a Cauchy sequence. Since $Y\subset X$ and $X$ is complete, there is a $\xi\in X$ such that
\begin{equation}
\lim_{n\to\infty} \mathfrak{d}(\xi_n,\xi)=0.
\end{equation}
Because $(\xi_n)$ is a sequence in $Y$ and $Y$ is closed, $\xi$ must be an element of $Y$. Thus, $Y$ is complete.
\end{proof}
\begin{theorem}\label{thm:completenessAppendix}
Define for each $R>0$
\begin{equation}
\P_R(\R^d):=\{\mu\in\P(\R^d):\supp\mu\subset B(R)\}.
\end{equation}
Fix $\nu_0\in\P_R(\R^d))$ and $T>0$, and define
\begin{equation}
\mathcal{C}:=\{\nu\in C([0,T];\P_R(\R^d)):\nu(0)=\nu_0\}.
\end{equation}
Then the following holds: endowed with the metric
\begin{equation}
\sup_{\tau\in[0,T]}\,\mathcal{W}(\mu_1(\tau),\mu_2(\tau)),
\end{equation}
the space $\mathcal{C}$ is a complete metric space.
\end{theorem}
\begin{proof}
Note that there is a one-to-one correspondence between elements of $\P_R(\R^d)$ and elements of $\P(B(R))$. Since Lemma \ref{lem:C0TPBRcomplete} states that
$C([0,T];\P(B(R)))$ is complete, the same must hold for $C([0,T];\P_R(\R^d))$, because convergence in one of these spaces implies convergence in the other. We omit further details.\\
Clearly, $\mathcal{C}\subset C([0,T];\P_R(\R^d))$. We now show that $\mathcal{C}$ is closed. Let $(\mu_n)\subset\mathcal{C}$ be a sequence that converges to $\mu\in C([0,T];\P_R(\R^d))$:
\begin{equation}
\lim_{n\to\infty}\sup_{t\in[0,T]}\mathcal{W}(\mu(t),\mu_n(t))=0.
\end{equation}
We note that
\begin{equation}
\mathcal{W}(\mu(0),\nu_0)=\mathcal{W}(\mu(0),\mu_n(0))\leqslant\sup_{t\in[0,T]}\mathcal{W}(\mu(t),\mu_n(t)).
\end{equation}
Since the left-hand side is independent of $n$, while the right-hand side tends to $0$ as $n\to\infty$,
\begin{equation}
\mathcal{W}(\mu(0),\nu_0)=0
\end{equation}
must hold, so $\mu(0)=\nu_0$. We conclude that $\mu\in\mathcal{C}$ and thus $\mathcal{C}$ is closed. It follows from Lemma \ref{lem:closedSubsetComplete} that $\mathcal{C}$ is complete.
\end{proof}

\end{appendix}

\end{document}